\documentclass[review]{siamart1116}
\usepackage{mathrsfs}
\usepackage{amsfonts}
\usepackage{color}
\usepackage{verbatim}
\usepackage[parfill]{parskip}
\newcommand{\off}{\color{black}}

\author{Lea Olja\v{c}a, Jochen Br\"{o}cker, Tobias Kuna\thanks{Lea Olja\v{c}a is funded by EPSRC Centre for Doctoral Training in Mathematics of Planet Earth, grant number: EP/L016613/1 and NCEO: National Centre for Earth Observations}}
\title{Almost sure error bounds for data assimilation in dissipative systems with unbounded observation noise}

\begin{document}

	\maketitle

\begin{abstract}
Data assimilation is uniquely challenging in weather forecasting due to the high dimensionality of the employed models and the nonlinearity of the governing equations. Although current operational schemes are used successfully, our understanding of their long-term error behaviour is still incomplete. In this work, we study the error of some simple data assimilation schemes in the presence of unbounded (e.g. Gaussian) noise on a wide class of dissipative dynamical systems with certain properties, including the Lorenz models and the 2D incompressible Navier-Stokes equations. We exploit the properties of the dynamics to derive analytic bounds on the long-term error for individual realisations of the noise in time. These bounds are proportional to the amplitude of the noise. Furthermore, we find that the error exhibits a form of stationary behaviour, and in particular an accumulation of error does not occur. This improves on previous results in which either the noise was bounded or the error was considered in expectation only. 
\end{abstract}


\section{Introduction}
\label{sec:Intro}
Data assimilation is a term used in the geophysical community to describe efforts to improve our knowledge of a system by combining incomplete observations with imperfect models~\cite{Apte2008}. Data assimilation is important in many fields of engineering and geophysical applications, and is an essential part of modern numerical weather prediction where it is used to initialise the forecasts based on observations of the atmosphere, combined with short term predictions~\cite{Kalnay2003}. In this field, data assimilation is uniquely challenging due to the infinite dimensionality and nonlinearity of the weather problem. Currently employed models use discretizations with \(\mathcal{O}(10^9)\) dimensional state vectors and \(\mathcal{O}(10^7)\) partial observations of the atmosphere per day~\cite{Bauer2015}. Furthermore, equations governing the dynamics of the atmosphere are well known to exhibit sensitive dependence on initial conditions~\cite{Lorenz1963,Kalnay2003}, meaning that determining them as accurately as possible is a key factor in increasing the length of the forecasting horizon.

Combining noisy data with uncertain models is an inverse problem whose optimal solution is necessarily probabilistic and sits naturally in a Bayesian framework~\cite{Law2012} and \cite{Kalnay2003}, Sec.~5.5. Due to the nonlinear nature of the underlying equations, deriving an explicit form for the posterior distribution is in general not possible~\cite{Stuart2010}. A sufficiently precise numerical representation (e.g.~by MCMC methods or particle filters~\cite{VanLeeuwen2009}) of the solution is very computationally expensive and not currently feasible in operational weather forecasting~\cite{Law2012}, although this is a promising area of research~\cite{VanLeeuwen2015}. Therefore, the data assimilation schemes used in practice are approximations based on exact schemes derived for linear systems with Gaussian priors and additive Gaussian noise, known as the Kalman filter~\cite{Law2015}. The schemes are applied to the nonlinear dynamics sequentially with various further simplifications, the simplest of which is to assume constant prior covariance. This is known as the 3DVAR method \cite{Kalnay2003}, Sec 5.5. A more advance method, the ensemble Kalman filter, involves an evolving prior covariance, estimated through the use of ensembles, that is, several simultaneous runs of the data assimilation cycle using a set of perturbed observations \cite{Evensen}. Although clearly used with great success~\cite{Bauer2015}, these are nonetheless ad hoc approximations, and a satisfactory understanding of their fundamental properties is still lacking.

A rigorous study of data assimilation in the context of the full primitive equations (a reasonable model of atmospheric circulation~\cite{Vallis}) is currently out of scope. There has been extensive study of a simpler but still infinite dimensional model; the 2D viscous, incompressible Navier-Stokes (N-S) equations. Other models typically studied in the context of data assimilation in geophysical applications (see e.g.~\cite{Lorenz1996,Law2014,Law2016,Sanz-Alonso2014}) are the Lorenz~'63 and Lorenz~'96 models, as they exhibit many of the properties of the \mbox{N-S} equations such as being dissipative with a quadratic and energy conserving nonlinearity, while having the advantage of being finite dimensional. Fortunately some remarkable properties of the 2D \mbox{N-S} equations have been known for some time. It was first shown by C.~Foias and G.~Prodi in 1967~\cite{Foias1967} that the solution is completely determined by the temporal evolution of some finite number of spatial Fourier modes, which have since been named the ``determining modes". Subsequent work~\cite{Foias1984,Jones1992} showed that this also holds for a finite set of appropriately chosen nodal values. 

More recent work re-frames these results in the context of data assimilation~\cite{Olson2003,Hayden2011}, and shows that certain data assimilation schemes have zero asymptotic error even with only finitely rank observations. Hayden, Olson and Titi~\cite{Hayden2011} consider the Lorenz~'63 and \mbox{N-S} equations with a data assimilation scheme where noiseless observations are directly replaced into the approximating solution at discrete times. Their result shows that for a sufficiently large number of observed low modes, the higher modes synchronise, that is, the error goes to zero with the number of assimilation cycles.

In~\cite{Brett2013}, Brett et al build on the results in~\cite{Hayden2011} by allowing for observational errors and using the 3DVAR algorithm. They show that for bounded observational errors, the asymptotic (\(t \to \infty\)) error between the approximating solution and the true state of the atmosphere is bounded, and of the same order of magnitude as the bound on the noise. The same result is obtained in \cite{Foias2016} for another type of data assimilation scheme, which is related to the once widely used 'nudging' schemes. Therefore, in both papers, the overall error is driven by the error in the observations, regardless of initial error. Furthermore, this result is obtained pointwise, that is, it is true for any realization of the noise. The stochastic properties of the observational errors however, do not enter into the derivation of the bound, except the boundedness, which is essential.

In~\cite{Law2014,Law2016}, results are obtained in expectation for unbounded noise for the Lorenz~'96 and~'63 models, respectively. They show that for the 3DVAR scheme, the mean square of the error is of the same order of magnitude as the variance of the noise. In~\cite{Sanz-Alonso2014}, Sanz-Alonso and Stuart extend this result, in expectation, to a wide class of dissipative PDEs, including infinite dimensional systems, that satisfy certain properties; the ``absorbing ball" property and the ``squeezing property''. As is noted in~\cite{Brett2013}, in a remark after Assumption 3.1, there is essentially a trade-off to be made between having bounded noise, with pointwise bounds, and unbounded noise, where similar techniques lead to results in expectation.

The main objective of the present paper is to investigate whether data assimilation into certain dissipative systems of PDEs is well behaved. Our approach is based on the works of~\cite{Hayden2011},~\cite{Brett2013} and~\cite{Sanz-Alonso2014}. In those publications, results regarding data assimilation accuracy with unbounded noise are given in expectation, while in the present paper we derive (almost surely) pointwise bounds, even for unbounded noise. More specifically, we prove that for large time, the error is bounded by a finite and stationary process, and give an explicit description of this process in terms of the observation noise. Technically, there are realisations of the noise for which this bound fails, but these have zero probability, and hence are statistically irrelevant. 

We use the \off simple replacement data assimilation scheme as studied by Titi et al in~\cite{Hayden2011} although we expect our result to be extendible to 3DVAR type algorithms as described in \cite{Brett2013}\off. We require assumptions similar to the absorbing and squeezing properties of~\cite{Sanz-Alonso2014} but with some crucial differences. We allow the squeezing function to be random, and require only that its expectation is less than one. We are then able to apply Birkhoff's Ergodic Theorem to show that the squeezing function is sufficiently often less than one to give us a bound which is pointwise finite (\cref{th:Main,th:Main2}). The result holds for any strength of the noise, given by the variance \(\sigma^2\), and furthermore, the bound decreases as the variance of the noise is decreased. Therefore the data assimilation error (for large time) is at least proportional to the strength of the noise. As in~\cite{Sanz-Alonso2014}, we test our assumptions on two finite dimensional systems; Lorenz~'63 and~'96, before turning to the infinite dimensional \mbox{N-S} system. 

The paper is organised as follows. In \cref{Sec:SetUp} we describe the dynamical system framework, the data assimilation scheme, and the assumptions we require on the observation error. Observations at each data assimilation time are assumed to contain a random error, the nature of which we keep as general as possible. In particular, we do not require i.i.d.\ or bounded noise, just that the noise is stationary and ergodic. In \cref{sec:Assump}, we set out general assumptions on the dynamical systems needed for our main result, \cref{th:Main}, the theorem itself and the proof. In \cref{Sec:Apriori}, we investigate the properties of an apriori bound we derive for the dissipative systems considered in this paper. In \cref{sec:Finite} we show that our assumptions are satisfied by a large class of finite dimensional dissipative systems provided they satisfy certain properties. We discuss the Lorenz~'63 and~'96 models as examples of such systems. In \cref{sec:NS}, we prove that the N-S equations satisfy the Assumptions of \cref{th:Main} as well. 

\section{The data assimilation problem}
\label{Sec:SetUp}

\subsection{Dissipative dynamical system}

Informally, we think of an equation as being ``dissipative" if all solutions are eventually bounded and this bound is uniform for any initial condition. Formally, a semigroup is dissipative if it possesses a compact absorbing set \cite{Robinson2001}.

Let \(\mathbf{H}\) be a Hilbert space with  \(| \ . \ |\) the induced norm. Let \(U\) be the solution of a dissipative system with initial conditions \(U_0\) at \(t_0\) and let \(\psi\) be the continuous semi-flow defined by 
\begin{equation}\label{eq:Sol}
U(t) = \psi(t,t_0,U_0), 
\end{equation}
where \[\psi(t+s,t_0,U_0) = \psi(t,s,\psi(s, t_0,U_0)), \ \ \ \ \text{(the semigroup property)},\] and \[\psi(0,t, U(t))=U(t) \] for all real \(t \geq 0\), such that \(\psi\) is continuous in \(t\) and with respect to initial condition \(U_0\).  

We assume that this dynamical system is a perfect representation of the real world system we are interested in; for instance the atmosphere, and we refer to $ U $ as the ``reference" solution.
 
\subsection{Data assimilation} 
As mentioned in the introduction, we will be using a simple data assimilation method as defined \off by Titi et al \off in \cite{Hayden2011} but with noise added at each discrete data assimilation time.

Let \(\mathbf{O}_P\), the observation \textit{space}, be a finite dimensional subspace of \(\mathbf{H}\) and \(P\) the orthogonal projection onto \(\mathbf{O}_P\).

An observation at time \(t_n\) is given by \(PU(t_n) + \sigma R_n\), where \(\sigma R_n\) is the noise, or random error, in the observation. We will define \(R_n\) more precisely in \cref{sec:Observations}. We assume that \(R_n\) is a random variable with values in \(\mathbf{O}_P\) so that \(PR_n = R_n\).

\off
We note that the observations as defined above are restricted to being finite in number and the observations space is restricted to a linear transformation of the model space. In weather prediction however, this is often not the case; the observation operator can be highly non-linear, as for example, in the case of satellite observations. Restricting to a linear observation operator also means that the additive nature of the noise is preserved.
\off

The approximating solution of the discrete data assimilation scheme that we use is obtained as follows. Initially at \(t_0 = 0\) we have, \[\bar{u}_0 = \eta + PU_0+ \sigma R_0,\] where \(\eta\) is the initial guess of the unobserved part of the solution. Then at discrete times \(0 < t_1 < t_2 < ...\)  we set
\begin{equation}\label{eq:Approx}
\bar{u}_{n} = Q\psi(t_{n},t_{n-1},\bar{u}_{n-1}) + PU(t_{n})+ \sigma R_{n},
\end{equation}
 where  $ Q=I-P $ is the projection onto unobserved space. 

At intermediate times \(t_n \leq t < t_{n+1}\), the approximating solution \(u(t)\) is a continuous in time function defined by 
\begin{equation}\label{eq:Approx_Sol}
u(t) = \psi(t,t_n,\bar{u}_n) \ \text{for} \ t \in [t_n,t_{n+1}).
\end{equation}
We note that \(u\) is continuous on each interval $ [t_n,t_{n+1}) $ but has discontinuities at \({t_n, n \in \mathbb{N}}\), with \(u\) continuous from the right and with limits to the left, since
\[u(t_n^+) = \lim_{t \to t_n^+}\psi(t,t_n,\bar{u}_n) = \bar{u}_n = u(t_n),\] while 
\[u(t_n^-) = \lim_{t \to t_n^-}\psi(t,t_{n-1},\bar{u}_{n-1}) = \psi(t_n,t_{n-1},\bar{u}_{n-1}) \neq \bar{u}_n.\]
We are interested in the data assimilation error \(\delta(t)\), which is the difference between the reference and approximating solutions described above. In particular, we are interested in the asymptotic behaviour as \(t \to \infty\). Like the approximating solution, \(\delta(t)\) is piece-wise continuous in time and defined by 
\begin{equation}\label{eq:error}
\delta(t) = U(t) - u(t) =\psi(t,t_0,U_0) - \psi(t,t_n,\bar{u}_n)
\end{equation}
  in the interval \([t_n,t_{n+1})\). At \(t_n\) we have \[\delta_n := \delta(t_n) = U(t_n) - \bar{u}_n =Q\psi(t_{n},t_0,U_0)-Q\psi(t_{n},t_{n-1},\bar{u}_{n-1}) - \sigma R_{n}. \] 
  For simplicity, we assume that the time between observational updates (the data assimilation interval), \[h = t_{n+1}-t_n >0 \] is constant.

\subsection{Observations}\label{sec:Observations}
As we will be considering the asymptotic data assimilation error, we will be looking at a sequence of noise realisation that extends into infinite time and in fact it will be useful to extend it backward in time also. 

 Let \((\Omega,\mathscr{F},\mathbb{P})\) be a probability space and \(T:\Omega \rightarrow \Omega\) a measure preserving map such that \(T\) and \(T^{-1}\) are ergodic with respect to $ \mathbb{P} $. Let  \(R:\Omega \rightarrow \mathbf{O}_P\) be a random variable on \((\Omega,\mathscr{F})\) and denote \(R_n = R \circ T^n\); a sequence of random variables, with \(n \in \mathbb{Z}\).  \(R_n\) will serve to model the noise in the observations at time \(t_n\).
We let \[\bar{R}:(\Omega, \mathscr{F}) \to (\mathbf{O}_P^{\infty},\mathscr{B}_{\infty}) \] be given by \[\omega \to (..R_{-1}(\omega), R_{0}(\omega),R_{1}(\omega)...).\] This is a measurable map and represents a realisation of the noise for all time, extending to infinite past and future. We denote the probability distribution of \(\bar{R}\) by \(P_{\bar{R}}\).

We note that with \(T\) measure preserving, \(R_n\) is a strictly stationary sequence (see e.g.~\cite{Breiman1992}, Proposition 6.9.\ for proof). We further assume that \(\mathbb{E}(R)\) = 0 and \(\mathbb{E}(|R|^2)=1\) and we model the random noise in our observation at time \(t_n\) as \(\sigma R_n\), where \(\sigma \in \mathbb{R^+}\). Therefore \(\sigma^2\) is the variance of the observation noise. If \(R\) were to have non-zero mean, this would represent a systematic error. 

As an example, suppose that the \(R_n\) are i.i.d random variables with \(T: \mathbf{O}_P^{\infty} \rightarrow \mathbf{O}_P^{\infty} \) being the shift map defined by \((T^k(\bar{r}))_n = r_{n+k}\) for \(\bar{r} \in \mathbf{O}_P^{\infty}\). Then the distribution \(P_{\bar{R}}\) of \(\bar{R}\) is the product probability and  \((\mathbf{O}_P^{\infty},\mathscr{B}_{\infty},P_{\bar{R}})\) is the canonical probability model\footnote{Since the distribution of the process contains all the information we are interested in, we have discarded the original process on \(\Omega\) and have represented it in term of the coordinate representation process instead on \(\mathbf{O}_P^{\infty}\). }. It can be shown that \(T\) is measure preserving and \(T\), \(T^{-1}\) are ergodic. The proof is similar to the Kolmogorov zero-one law~\cite{Breiman1992}, Theorem 3.12.
\section{Assumptions and main result}
\label{sec:Assump}
In this section we state the main assumptions that we will need in order to prove our main result, \cref{th:Main}. Assumption~\ref{as:one} requires the existence of an absorbing ball which is natural to dissipative systems. Assumption~\ref{as:two} can often be deduced from the same estimates that give us Assumption~\ref{as:one}, as is demonstrated in \cref{lemma:apriori}, and is an a priori bound on the error dynamics. Assumptions~\ref{as:three} and~\ref{as:four} are generally more difficult to prove, particularly Assumption~\ref{as:four} in the presence of unbounded random error. They represent a kind of contraction or squeezing on the unobserved part of the dynamics.\\

\newtheorem{myAs}{Assumption}
\newtheorem{myPrp}{Property}

\begin{myAs}\label{as:one}
\textbf{(Absorbing ball property)} There exists \(K > 0\), depending on the dynamical system, such that the ball \(\mathscr{B} = \{U; |U|^2\leq K\}\) is absorbing and forward invariant.
\end{myAs} 
\begin{myAs}\label{as:two} \textbf{(A priori bound)} 
For all \(\sigma, h >0\), there exists a measurable function \(\rho_0:\mathbb{R^+}\times \mathbb{R^+} \times \Omega \to \mathbb{R^+}\) with \[|\delta_n|^2 \leq  \rho_0(h,\sigma) \circ T^n(\omega):=\rho_n\] such that \(\rho_n\) is a continuous monotone increasing function of \(\sigma\).
\end{myAs} 

\begin{myAs} \label{as:three}
There exist continuous functions \(M, \gamma:(\mathbb{R^{+}}, \mathbb{R^{+}}) \rightarrow \mathbb{R^{+}}\) such that whenever \(U \in \mathscr{B}\) and \(|U-V|\leq \rho\), 
\[|Q\{\psi(t+\tau, t, U)-\psi(t+\tau, t, V)\}|^2 \leq M(\tau, \rho)|Q(U-V)|^2 + \gamma(\tau, \rho)|P(U-V)|^2.\] 
\end{myAs} 

\textbf{Remark}: Without loss of generality we can assume that \(M\) and \(\gamma\) are not decreasing in \(\rho\) because we can always replace \(M, \gamma\) by functions that are larger and not decreasing. 
\begin{myAs} \label{as:four}
With \(\rho_0\) as in Assumption~\ref{as:two} and \(M(\tau,\rho)\) and \(\gamma(\tau, \rho)\) as in Assumption~\ref{as:three}; for every \(\sigma>0\) there exists an \(h > 0\), such that \[\mathbb{E}M(h,\rho_0(h,\sigma)) < 1, \] and \[\mathbb{E}\gamma(h,\rho_0(h,\sigma)) < \infty.\]
\end{myAs} 

\textbf{Remark}: We note that for any measurable function \(f:\mathbb{R} \to \mathbb{R}\), the process \(f \circ \rho_n \) is stationary and ergodic, since \(T\) is assumed to be measure preserving and ergodic. 

In particular, we can write,
\begin{equation}
M_n(\tau):=M (\tau, \rho_n) = M_0(\tau, \rho_0) \circ T^n(\omega), \label{eq:MT} 
\end{equation}
and 
\begin{equation}
\gamma_n(\tau):=\gamma (\tau, \rho_n) = \gamma_0(\tau,\rho_0) \circ T^n(\omega). \label{eq:GammaT} 
\end{equation}

\label{sec:MainResult}
We now state the main result of the paper.
\begin{theorem}\label{th:Main} Suppose Assumptions~\ref{as:one} to~\ref{as:four} hold. Let $ \sigma^*>0 $ and take \(h>0\) as in Assumption~\ref{as:four} with \(\sigma^*\) instead of \(\sigma\). Then there exists a stationary and a.s.\ finite process \(C_n\), a \off non-negative \off constant \(\bar{\beta} < 1\) and a random variable \(D\), such that for all \(\sigma < \sigma^*\), the error \(\delta_n = U(t_n) - u(t_n)\) satisfies
\begin{equation}
|\delta_n|^2 \leq \sigma^2C_n+D\bar{\beta}^n|QU_0-\eta|^2,
\end{equation}
almost surely. In particular,
\begin{equation} \label{eq:MainEq}
\limsup_{n}\Big(|\delta_n|^2 - \sigma^2 C_n \Big)\leq 0,
\end{equation}
a.s., where \(C_n\), \(\bar{\beta}\) and \(D\) are given in the proof by Equations \eqref{eq:middle_term}, \eqref{eq:C_n} and~\eqref{eq:B_explicit}.
In particular, \(C_n, \bar{\beta} \) and \(D\) only depend on \(\sigma^*\).
\end{theorem}

\Cref{th:Main} shows that, for almost all realisations of the noise, at any data assimilation update time \(t_n\), the error \(\delta_n\) is bounded. In addition, asymptotically for large time, the bound is given by \(\sigma^2C_n\) which constitutes a stationary process so that its distribution is time independent. Furthermore as \(\sigma \to 0\) the bound decreases to zero like \(\sigma^2\).

To get a bound for intermediate times \(t \in (t_n, t_{n+1})\), we require a further assumption.
\begin{myAs}\label{as:five}
There exists a constant \(\kappa>0\) such that \(|\delta(t)|^2 \leq e^{\kappa(t-t_n)}|\delta_n|^2\) for \(t \in [t_n,t_{n+1}).\)
\end{myAs}

We can easily see that if Assumption~\ref{as:five} holds, then the following modified version of \cref{th:Main} follows.

\begin{theorem}\label{th:Main2} Suppose Assumptions~\ref{as:one} to~\ref{as:five} hold. Let $ \sigma^*>0 $ and take \(h>0\) as in Assumption~\ref{as:four} with \(\sigma^*\) instead of \(\sigma\). Then there exists a stationary and a.s.\ finite process \(C_n\), a \off non-negative \off constant \(\bar{\beta} < 1\) and a random variable \(D\), such that for all \(\sigma < \sigma^*\), the error \off \(\delta(t) = U(t) - u(t)\) with \(t \in [t_n, t_{n+1}):=I\)\off \ satisfies
\begin{equation}
|\delta(t)|^2 \leq (\sigma^2C_n+D\bar{\beta}^n|QU_0-\eta|^2)e^{\kappa h},
\end{equation}
almost surely. In particular, 
\off
\begin{equation} 
\limsup_{n}\Big[\sup_{t\in I}\Big(|\delta(t)|^2 - e^{\kappa h}\sigma^2 C_n \Big)\Big]\leq 0,
\end{equation}
\off 
a.s., where \(C_n\), \(\bar{\beta}\) and \(D\) are given in the proof by Equations \eqref{eq:middle_term}, \eqref{eq:C_n} and~\eqref{eq:B_explicit}.
In particular, \(C_n, \bar{\beta}\) and \(D\) only depend on \(\sigma^*\).
\end{theorem}

\begin{comment}
\begin{theorem}\label{th:Main2}
Under Assumptions~\ref{as:one} to~\ref{as:five}, there exists a stationary, a.s.\ finite process \(C_n\) such that \(\delta(t) = U(t)-u(t)\) satisfies \[\limsup_{n \to \infty}\Big(|\delta(t)|^2 - e^{\kappa h}\sigma^2 C_n \Big)\leq 0,\] with \(C_n\) as given in \cref{th:Main}.
\end{theorem}
\end{comment}

Before turning to the proof of the main result, we require some lemmas.
\begin{lemma} \label{Lemma:Meq}
Under Assumptions~\ref{as:one} to \ref{as:three}, \(\delta_n=U(t_n)-u(t_n)\) satisfies
\begin{equation}\label{eq:Meq}
|\delta_{n}|^2 \leq \sigma^2\sum_{l=1}^{n} \prod_{k=l}^{n-1}M_{k}|R_{l-1}|^2\gamma_{l-1} +\prod_{k=0}^{n-1}M_{k}|QU_0-\eta|^2 + \sigma^2|R_{n}|^2, 
\end{equation}
where \(M_k:= M(h,\rho_k(h))\) and \(h=t_{n+1} - t_n\) is the update interval.
\end{lemma}
\begin{proof}
 By Assumption~\ref{as:one} we have that the solution \(U(t) \in \mathscr{B} \) for some \(t>0\). Without loss of generality we can assume that \(U(t_0) \in \mathscr{B} \). Then, \(U(t_n) \in \mathscr{B} \), by the forward invariance of \(\mathscr{B}\). Furthermore, by Assumption~\ref{as:two},  for any \(h>0\), we have a stationary process \(\rho_n\) such that \(|\delta_n|^2 \leq \rho_n\) for all \(n \in \mathbb{N}\). Therefore we can apply Assumption~\ref{as:three} at each update time \(t_n\). Let  \(t \in [t_n,t_{n+1})\), \(U = U(t_n), V=u(t_n)\), and \(M_n(\tau)\)(respectively $ \gamma_n(\tau) $) be as in Equation~\cref{eq:MT} (respectively Eq.~\cref{eq:GammaT}) where \(\tau = t-t_n \in [0, h)\). We obtain
\off
\begin{align*}
|Q \delta_{n+1}|^2 & = \lim_{t \to t_{n+1}} |Q \delta(t)|^2  \\
& \leq \lim_{t \to t_{n+1}} M_n(t-t_n)|Q\delta_{n}|^2 + \sigma^2\gamma_n( t-t_n)|R_{n}|^2 \\
& = M_n(h)|Q\delta_{n}|^2+ \sigma^2\gamma_n( h)|R_{n}|^2 ,
\end{align*}
\off
where we have used the continuity of \off \(Q\delta(t)\) \off at \(t_{n+1}\). Write \(M_n: = M_n(h)\) and \(\gamma_n: = \gamma_n(h)\) for simplicity.
By induction on the above,
\iffalse
\begin{align*}
|Q\delta_{n}|^2 & \leq M_{n-1}|Q\delta_{n-1}|^2 + \sigma^2|R_{n-1}|^2 \gamma_{n-1}\\
 & \leq M_{n-1}M_{n-2}|Q\delta_{n-2}|^2 + M_{n-1}\sigma^2|R_{n-2}|^2\gamma_{n-2}+ \sigma^2 |R_{n-1}|^2\gamma_{n-1} \\
 & ...\\
 & \leq \prod_{k=0}^{n-1}M_{k}|Q\delta_{0}|^2 +\sigma^2\prod_{k=1}^{n-1}M_{k}|R_{0}|^2\gamma_0+ ...+\sigma^2\prod_{k=n-1}^{n-1}M_{k}|R_{n-2}|^2\gamma_{n-2}+\sigma^2|R_{n-1}|^2\gamma_{n-1}\\
 & = \sigma^2\sum_{l=1}^{n-1} \prod_{k=l}^{n-1}M_{k}|R_{l-1}|^2\gamma_{l-1} +\prod_{k=0}^{n-1}M_{k}|QU_0-\eta|^2 + \sigma^2|R_{n-1}|^2\gamma_{n-1},
\end{align*}
since \(|Q\delta_0|^2 = |QU_0- \eta|^2\).

Finally,
\begin{align*}
|\delta_n|^2 &=|Q\delta_n|^2+|P\delta_n|^2\\  &\leq \sigma^2\sum_{l=1}^{n-1} \prod_{k=l}^{n-1}M_{k}|R_{l-1}|^2\gamma_{l-1} +\prod_{k=0}^{n-1}M_{k}|QU_0-\eta|^2 + \sigma^2\gamma_{n-1}|R_{n-1}|^2 + \sigma^2|R_{n}|^2
\end{align*}
as required.
\fi
\begin{equation*}
|Q\delta_{n}|^2 \leq \sigma^2\sum_{l=1}^{n} \prod_{k=l}^{n-1}M_{k}|R_{l-1}|^2\gamma_{l-1} +\prod_{k=0}^{n-1}M_{k}|QU_0-\eta|^2,
\end{equation*}
since \(|Q\delta_0|^2 = |QU_0- \eta|^2\) and we define \(\prod_{k=n}^{n-1}M_{k}=1\).

Finally, using that \(|P\delta_n|^2 = \sigma^2|R_n|^2\),
\begin{align*}
|\delta_n|^2 &=|Q\delta_n|^2+|P\delta_n|^2,\\  &\leq \sigma^2\sum_{l=1}^{n} \prod_{k=l}^{n-1}M_{k}|R_{l-1}|^2\gamma_{l-1} +\prod_{k=0}^{n-1}M_{k}|QU_0-\eta|^2 + \sigma^2|R_{n}|^2,
\end{align*}
as required. \end{proof}
To obtain a meaningful bound as stated in \cref{th:Main}, we need that the RHS of estimate~\cref{eq:Meq} is almost surely finite in the long term. This would clearly be the case if \(M_k\) would be less than one, for all \(k\) (with some conditions on \(\gamma_n\)). Unfortunately, since the a priori bound is stochastic, the \(M_k\) are also stochastic and it is not, in general, possible to guarantee that \(M_k< 1\) for all \(k\), whatever the value of \(h\). However, we are able to use the Ergodic Theorem to show that if \(\mathbb{E}(M_k)<1\), it ensures \(M_k<1\) often enough to guarantee that estimate~\cref{eq:Meq} is almost surely finite. That is, for almost all realizations of the sequence \(\{M_k\}_k\), the proportion of \(M_k<1\) is sufficient to ensure that the product is less than 1.
\begin{lemma} \label{Lemma:Mprod}
For any real \(\xi >0\), there exist almost surely finite random variables \(C_{\omega, \xi}\) and \( C_{\omega, \xi}^{'}\), such that for all \(N>0\)
\begin{equation}\label{eq:Mprodpos}
\prod_{k=0}^{N-1} M_{-k} \leq C_{\omega, \xi}(\beta + \xi)^N,
\end{equation}
\begin{equation}\label{eq:Mprodneg}
\prod_{k=0}^{N-1} M_{k} \leq C_{\omega, \xi}^{'}(\beta + \xi)^N,
\end{equation}
where
\begin{equation}\label{eq:C}
C_{\omega, \xi}:=\max_{N} \frac{\prod_{k=0}^{N-1} M_{-k}}{(\beta + \xi)^N},
\end{equation} 
\begin{equation}\label{eq:C'}
C_{\omega, \xi}^{'}:=\max_{N} \frac{\prod_{k=0}^{N-1} M_{k}}{(\beta + \xi)^N},
\end{equation}
where \(\{M_k\}\) is as in \cref{Lemma:Meq} and 
\(\beta = \mathbb{E}(M_k)\).
\end{lemma}
\begin{proof}
Assuming \(\log{M_0(\omega)}\) is measurable we can apply the Ergodic Theorem~\cite{Walters, Breiman1992} to \(T^{-1}\) to obtain
\begin{align}\label{eq:ergodic_eq}
\lim_{n \to \infty} \frac{1}{n}\sum_{k=0}^{n-1} \log{M_{-k}(\omega)} & = \lim_{n \to \infty} \frac{1}{n}\sum_{k=0}^{n-1} \log{M_0(\omega) \circ T^{-k}(\omega)} \nonumber \\ & = \mathbb{E}(\log{M_0(\omega)}) \nonumber \\ & \leq \log{\mathbb{E}(M_0(\omega)),}
\end{align}
where the last inequality follows from Jensen's Inequality.

We note that we did not require that  \(\log{M_0(\omega)}\) is integrable as we can apply the Ergodic Theorem to random variables that are either bounded below or above. In the present case,  \(\log{M_{0}(h,\omega)}\) could be unbounded below but we may replace it with \(\bar{M_0}(h,\omega)=\max(\epsilon,M_0(h,\omega))\) for some small \(\epsilon > 0\) and apply the Ergodic Theorem to \(\log{\bar{M_0}(h,\omega)}\).

Let \(\beta =\mathbb{E}(M_k)\). From Inequality~\cref{eq:ergodic_eq} we have that for a.e.\ \(\omega\), for all \(\xi > 0,\) there exists \(\ N_{\omega,\xi}\) such that for all \(n\geq N_{\omega,\xi}\),
\[\frac{1}{n}\sum_{k=0}^{n-1} \log{M_{-k}} \leq \ln(\beta + \xi),\] and hence
\[ \prod_{k=0}^{n-1} M_{-k} \leq (\beta + \xi)^n.\] This implies
\begin{equation}\label{eq:Mprodfinite}
\frac{\prod_{k=0}^{n-1} M_{-k}}{(\beta + \xi)^n} \leq 1.
\end{equation}

Next, we note that for all \(N>0\) it holds that
\begin{align*}
\prod_{k=0}^{N-1} M_{-k} & = \frac{\prod_{k=0}^{N-1} M_{-k}}{(\beta + \xi)^N} (\beta + \xi)^N\\
 & \leq C_{\omega, \xi} (\beta + \xi)^N,
\end{align*}
where \[C_{\omega, \xi}:=\max_{N} \frac{\prod_{k=0}^{N-1} M_{-k}}{(\beta + \xi)^N}. \] 
 
 \(C_{\omega, \xi}\) is finite for a.e.\ \(\omega\) since by Inequality~\cref{eq:Mprodfinite} it is less than 1 for large enough \(N\). To get estimate~\cref{eq:Mprodneg}, we repeat the proof above with \(k=-k\) but using the ergodicity and \(\mathbb{P}\)-invariance of \(T\).   \end{proof}
 
 \begin{lemma} \label{Lemma:Mterm}
 Let \(\chi_n(\omega) = \chi_0 \circ T^{n}(\omega)\) be a sequence of random variables and let \[ E_{n,m}: =\sum_{l=m}^{n}\Big(\prod_{k=l}^{n}M_{k}\Big)\chi_l\] for \(n>m\).
 Then \[E_{n,0} = E_{0,-n} \circ T^{n}.\] 
 \end{lemma}
 
\begin{proof}
 \begin{align*}
 E_{n,0}(\omega) & =\sum_{l=0}^{n}\Big(\prod_{k=l}^{n}M_{k}\Big)\chi_l(\omega) \\
 & = \sum_{l=-n}^{0}\Big(\prod_{k=l}^{0}M_{k+n}\Big)\chi_{l+n}(\omega) \\
 & = \sum_{l=-n}^{0}\Big(\prod_{k=l}^{0}M_{k} \circ T^{n}(\omega)\Big) \Big(\chi_n \circ T^{n}(\omega)\Big)_l \\
 & = E_{0,-n} \circ T^{n} (\omega).
 \end{align*}
\end{proof}

 \begin{lemma}\label{Lemma:Mterm2} Let \(\chi_n(\omega) = \chi_0 \circ T^{n}(\omega)\) be non-negative random variables with finite expectation, and suppose Assumption \ref{as:four} holds. Then \[ \sum_{l=0}^{n} \prod_{k=l}^{n-1}M_{k}\chi_l \leq B_\xi \circ T^{n-1}, \]
  
  where
  \begin{equation}\label{eq:beta_xi}
   B_\xi= C_{\omega, \xi}\sum_{l=0}^{\infty} (\beta + \xi)^{l} \chi_{-l} + \chi_1
  \end{equation}
   is an almost surely finite random variable and \( C_{\omega, \xi}\) is as defined by~\cref{eq:C}.
  \end{lemma}
  
\begin{proof}
  By definition and by \cref{Lemma:Mterm}, we have that \[\sum_{l=0}^{n-1} \prod_{k=l}^{n-1}M_{k}\chi_l = E_{n-1,0} =  E_{0,-(n-1)} \circ T^{n-1} (\omega),\]
  where \[E_{0,-n} (\omega) = \sum_{l=-n}^{0}\Big(\prod_{k=l}^{0}M_{k}\Big) \chi_l = \sum_{l=-n}^{0}\Big(\prod_{k=0}^{-l}M_{-k}\Big)\chi_l. \] 
  Therefore,
  \[\sum_{l=0}^{n} \prod_{k=l}^{n-1}M_{k}\chi_l =E_{n-1,0} + \chi_n=\Big(E_{0,-(n-1)} +\chi_1\Big)\circ T^{n-1} (\omega).\]
  
  Then using estimate~\cref{eq:Mprodpos} from \cref{Lemma:Mprod} we have
  \begin{align*} 
  E_{0,-(n-1)} & \leq C_{\omega, \xi}\sum_{l=-(n-1)}^{0} (\beta + \xi)^{|l|} \chi_l \nonumber \\
   & \leq C_{\omega, \xi} \sum_{l=-\infty}^{0} (\beta + \xi)^{|l|} \chi_l,\\
   &= C_{\omega, \xi}\sum_{l=0}^{\infty}  (\beta + \xi)^{l} \chi_{-l}.
  \end{align*}
 %& \leq \sum_{l=-1}^{\infty} C_{\omega, \xi} (\beta + \xi)^{l} \chi_{-l},where the last inequality hold because \(C_{\omega, \xi} \geq 1\) and \((\beta + \xi)^{-l}> 1\)?? %
Let  \[B_\xi :=C_{\omega, \xi}\sum_{l=0}^{\infty} (\beta + \xi)^{l} \chi_{-l}+\chi_1,\]
then 
  \[\sum_{l=0}^{n} \prod_{k=l}^{n-1}M_{k}\chi_l \leq B_\xi \circ T^{n-1}, \] as required.
  
  It is clear that \(B_\xi\) is measurable since \(\chi_n\) are non-negative. We need to show that \(B_\xi\) is finite for a.e.\ \(\omega\).
  
  Since \(\beta < 1\) by Assumption~\ref{as:four}, we can choose \(\xi>0\) such that \(\beta + \xi < 1\). We know that \(C_{\omega, \xi}\) is a.s.\ finite by \cref{Lemma:Mprod} and \(\chi_1\) is non-negative with finite expectation. Hence, by Monotone Convergence Theorem, \[\mathbb{E}\Big(\sum_{l=0}^{\infty} (\beta + \xi)^{l} \chi_{-l}\Big) = \sum_{l=0}^{\infty} (\beta + \xi)^{l} \mathbb{E}(\chi_{-l}) < \infty. \] Hence \[\sum_{l=0}^{\infty} (\beta + \xi)^{l} \chi_{-l} < \infty\] almost surely. Therefore, \(B_\xi\) is a.s.\ finite as required. 
\end{proof}
    
For clarity, where necessary, we will use \(\sigma\) as a parameter in the notation for the remainder of this section.
\begin{proof}[Proof of \cref{th:Main}] 

By our choice of  $ \sigma^* $ and $ h $, we have $ \mathbb{E}M^*_k < 1 $, where \\ \(M_k^*:=~M_k(h,\rho_k(h,\sigma^*))\).

We consider Inequality~\cref{eq:Meq}. By monotonicity of \(M_k\) and \(\gamma_k\) we can replace \(\sigma^*\) inside the functions so that the Inequality~\cref{eq:Meq} still holds. We have
\begin{equation}\label{eq:Meq_sigma_star}
|\delta_{n}|^2 \leq \sigma^2\sum_{l=1}^{n} \prod_{k=l}^{n-1}M_{k}^*|R_{l-1}|^2\gamma_{l-1}^* +\prod_{k=0}^{n-1}M_{k}^*|QU_0-\eta|^2 + \sigma^2|R_{n}|^2, 
\end{equation}
where \(M_k^*:= M(h,\rho_k(h,\sigma^*))\) and \(\gamma_k^*:= \gamma(h,\rho_k(h,\sigma^*))\).

We note first that the second term of Inequality~\cref{eq:Meq_sigma_star} is bounded a.s.\ by~\cref{eq:Mprodneg};
\[\prod_{k=0}^{n-1}M_{k}^*|QU_0-\eta|^2 \leq C_{\omega, \xi}^{'*}(\beta^{*} + \xi)^{n} |QU_0-\eta|^2,\] where \(\beta^* =\beta(\sigma^*)\) and \(C_{\omega, \xi}^{'*} = C_{\omega, \xi}^{'}(\sigma^*).\)
Fix \(\xi>0\) so that \(\beta^* + \xi < 1\). Then,
\begin{equation}\label{eq:middle_term}
\lim_{n \to \infty}\prod_{k=0}^{n-1}M_{k}^*|QU_0-\eta|^2 \leq \lim_{n \to \infty} D\bar{\beta}^{n} |QU_0-\eta|^2 =0,
\end{equation}
with \(D=C_{\omega, \xi}^{'*}\) and \(\bar{\beta} = \beta^* +\xi\). 

Next, we use \cref{Lemma:Mterm2}. Let 
\begin{equation}\label{eq:C_n}
C_n: = B_\xi^* \circ T^{n-1} + |R\circ T^n|^2,
\end{equation} where \(B_\xi^*\) is as defined by Equation~\cref{eq:beta_xi} with \(\sigma\) replaced by \(\sigma^*\) and \(\chi_l = |R_{l-1}|^2\gamma_{l-1}^*\). Hence explicitly, 
\begin{equation}\label{eq:B_explicit}
 B_{\xi}^{*}= C_{\omega, \xi}^*\sum_{l=0}^{\infty} \bar{\beta}^{l} |R_{-l+1}|^2\gamma_{-l+1}^* + |R_{0}|^2\gamma_{0}^*,
\end{equation}
with \(\bar{\beta} = \beta^* + \xi\). The remaining terms of Inequality~\cref{eq:Meq_sigma_star} are bounded by \(\sigma^2 C_n\) which is a.s.\ finite and stationary by \cref{Lemma:Mterm2} and by our assumptions on \(R_n\).

Therefore,
\[|\delta_n|^2 \leq \sigma^2C_n + D\bar{\beta}^n|QU_0-\eta|^2\] and
\[\limsup_{n}\Big(|\delta_n|^2 -\sigma^2 C_n\Big) \leq 0,\] by Equation \eqref{eq:middle_term} a.s.\ as required. Furthermore it holds that \(\sigma^2C_n \to 0\) as \(\sigma \to 0\) since \(C_n\) does not depend on \(\sigma\).  \end{proof}

\section{A priori bound for strongly dissipative systems}\label{Sec:Apriori}

The next lemmas show that we can usually have a more explicit candidate for the a priori bound \(\rho_n\), if one has an estimate of the rate of contraction to the attractor. This rate is closely related to the absorbing ball property and to our requirement that the system is dissipative. This contraction can be shown to hold for many important dynamical systems, such as Lorenz~'63, '96 and the 2D, incompressible, Navier-Stokes. In fact, it is how we are able to show that these systems have the absorbing ball property and are dissipative. We will study this in more detail in the subsequent sections.

The next lemma derives a bound on the approximating solution based on a specific rate of contraction. The bound depends on the observation noise up to time \(t_n\), the initial guess \(\eta\), initial condition \(U(t_0)\) and the length of the data assimilation interval \(h\).

\begin{lemma} \label{lemma:apriori}
Let $ U $ be a solution to a semi-dynamical system and suppose that there exist constants \(c_1, c_2 >0\) such that 
\begin{equation} \label{eq:dissip}
|U(t)|^2 \leq e^{-c_1 (t-s)}|U(s)|^2 + c_2
\end{equation}
 for all \(0 \leq s < t\). 
 Let \(u(t)\) be the approximating solution as defined by Equation~\cref{eq:Approx_Sol}, then
 \off
  \begin{equation} \label{eq:iterated_dissip}
 |u(t_n)|^2 \leq  \phi_n(h,\eta, |U(t_0)|^2)+ 2\sigma^2\sum_{k=0}^{n} e^{-c_1kh}|R_{n-k}|^2
  \end{equation}
  \off
 for all \(n \in \mathbb{N}\), where \(h=t_n-t_{n-1} \) and 
 \off
 \[\phi_n(h,\eta,x) = |\eta|^2+\frac{2x}{c_1h}+ 3c_2\frac{1- e^{-c_1nh}}{1- e^{-c_1h}}.\]
 \off
\end{lemma}

\begin{proof} By Inequality~\cref{eq:dissip} and because \(u_{n-1}(t)\) is a solution in the interval \([t_{n-1},t_n)\), we have
\begin{equation} \label{eq:approx_dissip}
|u(t_n^-)|^2 \leq e^{-c_1h}|u(t_{n-1})|^2 + c_2.
\end{equation}
By definition and continuity of \(Qu(t)\) at \(t_{n}\) we have
\begin{equation}\label{eq:PQ_sept}
|u(t_{n})|^2 = |Qu(t_{n}^-)|^2 + |PU(t_{n})+\sigma R_{n}|^2 \leq |u(t_{n}^-)|^2 + |PU(t_{n})+\sigma R_{n}|^2.
\end{equation}
For simplicity, let \(O_n =|PU(t_n)+\sigma R_{n}|^2 \) and substitute Inequality~\cref{eq:approx_dissip} into Inequality~\cref{eq:PQ_sept} to get;
\begin{equation*}
|u(t_n)|^2 \leq e^{-c_1h}|u(t_{n-1})|^2 + O_{n-1}  + c_2.
\end{equation*}
Therefore by induction
\begin{equation} \label{eq:iterative_interm}
|u(t_n)|^2 \leq  e^{-c_1nh}|u(t_0)|^2+ \sum_{k=0}^{n-1} e^{-c_1kh}\Big(O_{n-k}+ c_2\Big).
\end{equation}
We note that
\begin{align*}
|PU(t_{n-k})+\sigma R_{n-k}|^2 &\leq 2|U(t_{n-k})|^2+2\sigma^2|R_{n-k}|^2\\ &\leq  2e^{-c_1(n-k)h}|U(t_0)|^2 +2c_2+ 2\sigma^2|R_{n-k}|^2,
\end{align*}
where we have used Inequality~\cref{eq:dissip} on \(U(t_{n-k})\). This implies
\begin{align*}
\sum_{k=0}^{n-1} e^{-c_1kh}O_{n-k} &\leq \sum_{k=0}^{n-1} e^{-c_1kh}\Big( 2e^{-c_1(n-k)h}|U(t_0)|^2 +2c_2+ 2\sigma^2|R_{n-k}|^2\Big)\\
&=2n e^{-c_1nh}|U(t_0)|^2 + 2c_2\frac{1- e^{-c_1nh}}{1- e^{-c_1h}} +2\sigma^2\sum_{k=0}^{n-1} e^{-c_1kh}|R_{n-k}|^2.
\end{align*}
Then Inequality~\cref{eq:iterative_interm} becomes
\begin{equation*} 
|u(t_n)|^2 \leq |\eta|^2+\frac{2|U(t_0)|^2}{c_1h}+ 3c_2\frac{1- e^{-c_1nh}}{1- e^{-c_1h}}+ 2\sigma^2\sum_{k=0}^{n} e^{-c_1kh}|R_{n-k}|^2.
\end{equation*}
where we have used that \(ne^{-c_1hn} \leq \frac{1}{c_1h}\) for all \(n\geq 0\) and \(h>0\) and  \(|u(t_0)|= |\eta|^2 + \sigma^2 |R_{0}|^2\), where \(\eta\) is the initial guess. Thus we have shown Inequality~\cref{eq:iterated_dissip}.
\end{proof}

We can readily see that Inequality~\cref{eq:dissip} gives us an absorbing ball \(B(0,r)\) with \(r>c_2^{1/2}\) since any bounded set will eventually be inside the ball. However, we cannot deduce forward invariance. We will see that the actual contractions we encounter in the dynamical systems we study, do guarantee forward invariance and hence imply that Assumption~\ref{as:one} holds.

% \(B(0,r) = {u:|U| \leq r}\) so that |U|^2 \leq r^2.%

The following corollary of \cref{lemma:apriori} gives the a priori bound required for Assumption~\ref{as:two}.
\begin{corollary} \label{cor:apriori}
 Let the conditions of \cref{lemma:apriori} hold and let \(\delta_n=U(t_n)-u(t_n)\) be the data assimilation error and \(h=t_{n}-t_{n-1}\) the update interval. Then there exists a stationary, a.s.\ finite process 
\begin{equation}\label{eq:apriori}
\rho_n = \bar{K}+ F(h)+ 4\sigma^2\sum_{k=0}^{\infty} e^{-c_1kh}|R_{n-k}|^2,
\end{equation}
 such that \(|\delta_n|^2 \leq \rho_n\), for all \(n \in \mathbb{N}\). 
\end{corollary}

\begin{proof}
 By definition of \( |\delta_n|^2\), we have
 \begin{equation}\label{eq:error_decomp}
 |\delta_n|^2 
 \leq 2|U(t_n)|^2+ 2|u(t_{n})|^2.
 \end{equation}
We insert \cref{eq:dissip,eq:iterated_dissip} into \cref{eq:error_decomp} to obtain
\begin{equation*}
|\delta_n|^2 \leq  2\phi(h,\eta,|U(t_0)|^2)+ 4\sigma^2\sum_{k=0}^{n} e^{-c_1kh}|R_{n-k}|^2 +2e^{-c_1hn}|U(t_0)|^2+2c_2.
\end{equation*}
The above simplifies to 
\[|\delta_n|^2 \leq  \bar{K}+ F(h)+ 4\sigma^2 \sum_{k=0}^{\infty} e^{-c_1kh}|R_{n-k}|^2,\]
where \(F(h) = \frac{6c_2}{1-e^{-c_1h}}+\frac{4|U(t_0)|^2}{c_1h} \) and \(\bar{K} = 2\Big(|U(t_0)|^2+ c_2+|\eta|^2\Big),\) as required.

To see that \(\rho_n\) is a measurable process, set \[\rho_n^N:= \bar{K}+ F(h) + 4\sigma^2\sum_{k=0}^{N} e^{-c_1kh}|R_{n-k}|^2.\] For each \(N\), $\rho_n^N $ is a finite sum of random variables and therefore measurable and $ \{ \rho_n^N\} $ is a pointwise non decreasing sequence, since we are adding non-negative terms. Therefore, \(\rho_n = \sup_{N} \rho_n^N,\) is measurable. To see that \(\rho_n\) is almost surely finite, we note that by the Monotone Convergence Theorem 
\begin{equation}\label{eq:exp_rho}
\mathbb{E}(\rho_n) = \sup_N \mathbb{E}(\rho_n^N) = \bar{K}+ F(h)+  \frac{4\sigma^2}{ 1-e^{-c_1h}} < \infty
\end{equation}
 for all \(h>0\).
Furthermore, \(\rho_n\) is stationary as \(R_n\) is stationary.
\end{proof}

We can see from Equation~\cref{eq:exp_rho} that the a priori bound behaves badly at \(h=0\) as its expectation is \(\mathcal{O}(\frac{1}{h})\), for small \(h\). In the next lemma we show that for almost all \(\omega \in \Omega\), \(\lim_{h \to 0} \rho_nh: = D_{\omega}\) exists. Therefore, pointwise, for small \(h\), \(\rho_n = \mathcal{O}(\frac{1}{h})\) as well. We note also that \(\rho_n\) is decreasing if the noise level \(\sigma\) decreases and converges to a noise-independent constant when \(\sigma \to 0\).

\begin{lemma}\label{Lemma:apriori}
For \(\rho_n\) as defined by Equation~\cref{eq:apriori} we have that
\begin{enumerate}
\item \(\lim_{h \to 0} \mathbb{E}(\rho_n)h =C <\infty\) where \(C>0\) is a constant,
\item \(\lim_{h \to 0} \rho_n(\omega)h = D_{\omega}\) for a.e. \(\omega\),
\item for all \(h>0\), \(\rho_n(\omega)\) is monotone in \(\sigma\) and $ \lim_{\sigma \to 0}\rho_n(\omega) =\bar{K} + F(h)$ almost surely. 

\end{enumerate}
\end{lemma}

\begin{proof}
To prove item 1, note that
\begin{align*}
\lim_{h \to 0} \mathbb{E}(\rho_n(\omega))h &= \lim_{h \to 0} (\bar{K}+ F(h)+ 
4\sigma^2 \sum_{k=0}^{\infty} e^{-c_1kh})h\\
&=\lim_{h \to 0} \frac{6c_2h}{1-e^{-c_1h}}+\frac{4|U(t_0)|^2h}{c_1h} + \frac{4\sigma^2h }{1-e^{-c_1h}}\\
&= \frac{6c_2+4|U(t_0)|^2+4\sigma^2}{c_1}:=C.
\end{align*}

To prove item 2, it remains to check the pointwise limit of the third term in Equation~\cref{eq:apriori}. Using summation by parts, for any \(N>0\),
\begin{equation}\label{eq:Parts}
\sum_{k=0}^{N} e^{-c_1kh}|R_{n-k}|^2 = e^{-Nc_1h} \sum_{k=0}^{N}|R_{n-k}|^2+\sum_{k=0}^{N-1}e^{-kc_1h}(1-e^{c_1h})\sum_{j=0}^{k}|R_{n-j}|^2.
\end{equation}

Considering the first term of RHS of Equation~\cref{eq:Parts}, by ergodicity of \(R_n\),
\begin{equation*}
\lim_{N \to \infty} Ne^{-Nc_1h} \frac{\sum_{k=0}^{N}|R_{n-k}|^2}{N} = \lim_{N \to \infty} \Big(Ne^{-Nc_1h}\Big)\mathbb{E}(|R_{n-k}|^2) =  0,
\end{equation*}
for a.e.\ \(\omega\).

Next we consider the second term. Again from ergodicity, we have that\\ \(\lim_{k \to \infty}\frac{\sum_{j=0}^{k}|R_{n-j}|^2}{k} = 1\), since \(\mathbb{E}(|R_n|^2) = 1\). Therefore, for any \(\epsilon>0\), there exists \(N_{\omega, \epsilon}\) such that for all \(k\geq N_{\omega, \epsilon}\), \(\frac{\sum_{j=0}^{k}|R_{n-j}|^2}{k} < 1+\epsilon\). Hence for any $ k>0 $, \[\sum_{j=0}^{k}|R_{n-j}|^2 =\frac{\sum_{j=0}^{k}|R_{n-j}|^2}{k}k \leq \bar{D}_{\omega}k, \] where \[\bar{D}_{\omega} := \sup_{k} (\frac{\sum_{j=0}^{k}|R_{n-j}|^2}{k}),\] and \(\bar{D}_{\omega} < \infty\) since for large enough \(k\) it is smaller than \(1+ \epsilon\).

Thus the second term of the RHS of Equation~\cref{eq:Parts} is bounded a.s.\ by \[ (1-e^{-c_1h})\bar{D}_{\omega} \sum_{k=0}^{N-1}e^{-kc_1h}k=(1-e^{-c_1h})\bar{D}_{\omega}\frac{e^{-c_1h}}{(1-e^{-c_1h})^2}=\bar{D}_{\omega}\frac{e^{-c_1h}}{(1-e^{-c_1h})}.\]
In summary, in the limit \(h \to 0\), \[\rho_n h  \to \frac{6c_2}{c_1}+\frac{4|U(t_0)|^2}{c_1}+ 4\sigma^2\bar{D}_{\omega}:=D_{\omega}\] and \(\rho_n = \mathcal{O}(\frac{1}{h})\) a.s.\ as required.

For item 3, we note that the random term of \(\rho_n\) is a.s.\ finite, therefore for a.e.\ \(\omega\), and \(h>0\), $ \lim_{\sigma \to 0}\rho_n =\bar{K}+ F(h)$, is a constant that does not depend on the noise. 
\end{proof}
 
\section{Application to finite dimensional systems}
\label{sec:Finite}
In this section we derive more concrete properties, sufficient to imply the general Assumptions~\ref{as:one}~to~\ref{as:five} in \cref{sec:Assump}, for dissipative and finite dimensional systems of the form 
\begin{equation}\label{eq:ODE}
\frac{dU}{dt} + A U + B(U,U) = f,
\end{equation}
 where solutions \(U\) and forcing \(f\) are functions in a finite dimensional vector space \(\mathbf{H}=\mathbb{R}^d\), $ A $ is a linear operator and \(B\) is a symmetric, bilinear operator; consequently, the results of \cref{th:Main,th:Main2} hold. In \cref{ss:L63,ss:L96} we apply our results to the Lorenz~'63 and Lorenz~'96 models  respectively.

We assume the following properties,
\begin{myPrp}\label{Prop:Finite}
\begin{enumerate}
\item $ B(U,V) = B(V,U)$ for all \(U,V \in \mathbf{H}\).
\item \((B(U,U),U) = 0,\) for all \(U \in \mathbf{H}\).
\item $ B(QU,QU)=0, $ for all \(U \in \mathbf{H}\).
\item  There exists a constant \(a_1>0\) such that for all \(U,V \in \mathbf{H}\), \\ \(|(B(U,V)| \leq a_1|U||V| \).
\item  \((AU,U) \geq |U|^2\), for all \(U \in \mathbf{H}\).
\end{enumerate}
\end{myPrp}
Similar properties are used in~\cite{Law2016},~\cite{Law2014} and~\cite{Sanz-Alonso2014}. For the Lorenz~'63 model and standard observation operator \(P\), as specified in \cref{ss:L63}, Properties~\ref{Prop:Finite}.1 to~\ref{Prop:Finite}.4 are easily deduced, while Property~\ref{Prop:Finite}.5 is shown in e.g.~\cite{Hayden2011}. For the Lorenz~'96 system and standard \(P\), as specified in \cref{ss:L96}, all the properties are shown in~\cite{Law2016}.

\textbf{Remark 1:} 
Property  \ref{Prop:Finite}.1 is not a restriction on our dynamical system \cref{eq:ODE} since only the symmetric part of $ B $ enters the dynamics anyway. Property \ref{Prop:Finite}.2 implies that the non-linear term does not contribute to the change in energy, analogous with the nonlinear part of the Navier-Stokes Equations. Property \ref{Prop:Finite}.3 effectively represents a non trivial condition on the observation operator $ P $, ensuring a form of observability of the system. Property \ref{Prop:Finite}.4 is true for any bilinear operator on a finite dimensional space and hence represents no loss of generality. Property \ref{Prop:Finite}.5 reflects the fact that \(Au\) is considered to be a dissipative term in the dynamics.

\textbf{Remark 2:} From the above description of the dynamical system, it is clear there are many parallels with the \mbox{N-S} equations, such as dissipativity, and a nonlinearity which is  quadratic and energy conserving. Furthermore, we will see in \cref{sec:NS} that the \mbox{N-S} equations can be rewritten in a very similar form as Equation \cref{eq:ODE}.

\textbf{Remark 3:} We note that by orthogonality of \(Q\) and following from Property~\ref{Prop:Finite}.5 we always have that
\begin{equation} \label{eq:Property6}
(AU,PU) \geq a_2|PU|^2-a_3|U|^2
\end{equation}
for some \(a_2 > 0\) and \(a_3\geq 0\).\footnote{\((AU,PU) = (A(P+Q)U,PU) = (APU,PU)+(AQU,PU)\geq |PU|^2 - \|A\||U|^2\), where we have used Property~\ref{Prop:Finite}.5. Therefore we have \(a_2 = 1\) and \(a_3 = \|A\|\) but these are not necessarily the sharpest such constants.}

\textbf{Remark 4:} We note that if Property~\ref{Prop:Finite}.3 holds for an orthogonal projection \(Q\) then they also hold for any projection whose image is contained in the image of \(Q\). 

The next two lemmas follow directly from Property~\ref{Prop:Finite}. For the case of Lorenz~'96, the proofs are given in~\cite{Law2014}.

\begin{lemma}\label{Lemma:Pre_Prop1.3}
Properties~\ref{Prop:Finite}.1~and~\ref{Prop:Finite}.2 imply that
\[(B(V,V),U) = -2(B(U,V),V)\] holds for all \(U,V \in \mathbf{H}\).
\end{lemma}
\begin{comment}
\textbf{Proof:} The proof follows by expanding \((B(U+V,U+V),U+V)\) and \((B(U-V,U-V),U-V)\) using bilinearity and Properties~\ref{Prop:Finite}.1~and~\ref{Prop:Finite}.2.
\end{comment}

The proof is simply expanding $ (B(U+V,U+V),U+V)$ and $ (B(U-V,U-V),U-V) $ using Properties~\ref{Prop:Finite}.1 and~\ref{Prop:Finite}.2 and bilinearity of \(B\).
\begin{comment}
\textbf{Proof:} 
By Properties~\ref{Prop:Finite}.1 and~\ref{Prop:Finite}.2 and bilinearity of \(B\), we obtain 
\begin{align*}
(B(U+V,U+V),U+V) = & 2(B(U,V),U)+2(B(U,V),V)\\&+(B(U,U),V)+(B(V,V),U) \\&= 0.
\end{align*}
Similarly, one computes that
\begin{align*}
(B(U-V,U-V),U-V) = &-2(B(U,V),U)+2(B(U,V),V)\\&-(B(U,U),V)+(B(V,V),U)  \\&= 0.
\end{align*}
Adding the two we get \[4(B(U,V),V)+2(B(V,V),U) = 0,\] which implies
\[(B(V,V),U)=-2(B(U,V),V),\] as required.\\ \qed 
\end{comment}

\begin{lemma}\label{Lemma:Prop1.3}
Suppose that Properties~\ref{Prop:Finite}.1,~\ref{Prop:Finite}.2 and \ref{Prop:Finite}.4 are satisfied. Then Property~\ref{Prop:Finite}.3  is equivalent to the following; there exists a constant \(b>0\) such that 
\begin{equation}\label{eq:New_prop}
2|(B(U,V),V)|\leq b|PV||U||V|.
\end{equation}
\end{lemma}

\begin{proof} By \cref{Lemma:Pre_Prop1.3}, \(2|(B(U,V),V)|=|(B(V,V),U)|\). Note that 
\begin{align*}
(B(V,V),U) &= (B(PV+QV,PV+QV),U)\\ &= 2(B(PV,QV),U)+(B(PV,PV),U),
\end{align*}
where we have used Property~\ref{Prop:Finite}.3. Therefore by Property~\ref{Prop:Finite}.4,
\begin{align*}
|(B(V,V),U)| &\leq 2a_1|PV||QV||U|+a_1|PV|^2|U|\\
&=a_1|PV||U|(2|QV|+|PV|)\\
&\leq 3a_1|PV||U||V|,
\end{align*}
as required with \(b=3a_1\).

Conversely, suppose that Inequality~\cref{eq:New_prop} holds. Then \[|B(QV,QV),U)|\leq b|PQV||QV||U| = 0\] since $ |PQV|=0 $. As this holds for all \(U \in \mathbf{H}\) we get that \(B(QV,QV)=0\) for all \(V \in \mathbf{H}.\)
\end{proof}

In the next several lemmas we show that if Property~\ref{Prop:Finite} holds, then ODEs of the form~\cref{eq:ODE} satisfy Assumptions~\ref{as:one} to~\ref{as:five}, and consequently \cref{th:Main,th:Main2} hold. 

We start with showing that Properties~\ref{Prop:Finite}.2 and~\ref{Prop:Finite}.5 imply Assumptions~\ref{as:one}~and~\ref{as:two}.
\begin{lemma}\label{Lemma:Ass12}
Let \(U\) be the solution of a finite dimensional ODE as defined by \cref{eq:ODE} and suppose that Properties~\ref{Prop:Finite}.2~and~\ref{Prop:Finite}.5 are satisfied. Then Assumption~\ref{as:one} holds for any \(K > |f|^2\) and Assumption~\ref{as:two} for \(\rho_n\) as given in \cref{cor:apriori} with \(c_1 = 1\) and \(c_2 = |f|^2\).
\end{lemma}

\begin{proof}
The absorbing ball property is easily verified. Take the inner product of ODE~\cref{eq:ODE} with \(U\) and use Property~\ref{Prop:Finite}.2 and Property~\ref{Prop:Finite}.5 to get
\[\frac{1}{2}\frac{d|U|^2}{dt}+|U|^2\leq (f,U). \] 
Then, by the Cauchy-Schwarz and Young's inequality we obtain
\[\frac{1}{2}\frac{d|U|^2}{dt}+|U|^2 \leq |(f,U)| \leq |f||U| \leq \frac{1}{2}|f|^2 + \frac{1}{2}|U|^2,\] and hence,
\[\frac{d|U|^2}{dt}+|U|^2\leq |f|^2. \]

Assumption~\ref{as:one} follows from using Gronwall's lemma;

\begin{equation}\label{eq:dissipF}
|U(t)|^2 \leq |U(0)|^2 e^{-t}+|f|^2(1-e^{-t}).
\end{equation}

We see that any ball \(B(0,K^{1/2})\) with \(K>|f|^2\) is absorbing and forward invariant. Furthermore, Inequality~\cref{eq:dissipF}  implies that the conditions of \cref{cor:apriori} are satisfied with \(c_1 = 1\) and \(c_2 = |f|^2\) and hence Assumption~\ref{as:two} (a priori bound) holds.
\end{proof}

Before proceeding to the next lemmas we derive an equation for the error \(\delta~=~U~-~u\). Since the approximating solution \(u\) satisfies Equation~\cref{eq:ODE} in the interval \([t_n,t_{n+1})\), we have that
\begin{equation}\label{eq:error_eq}
\frac{d\delta}{dt} + A \delta + 2B(U,\delta)- B(\delta, \delta) = 0,
\end{equation}
where we have used the bilinearity and symmetry of \(B\) to derive the above.
 
In the next Lemma we show that Assumption~\ref{as:five} holds (Eq. \cref{eq:er_bound}), and we derive a bound on \(|P\delta|\) (Eq. \cref{eq:Pdelta_bound}) which is used in \cref{Lemma:ass3} to show that Assumption~\ref{as:three} holds. The bound on \(|P\delta|\) and its proof are similar to that of the bound obtained in \cite{Sanz-Alonso2014}, Lemma 5.3, but with an important difference. If we were to simply replace the bound on \(|\delta_0|\) (given by \(r'^2\) in that paper) by our a priori bound \(\rho_n\), we would have a term multiplying \(|\delta|^2\) that in the limit \(h \to 0\) tends to a constant (see \cref{Lemma:apriori}). In our bound \cref{eq:Pdelta_bound}, the a priori bound appears in lower order, \(\rho_n^{1/2}\). This means that in the limit, this term goes to zero, which, in turn, enables us to show in \cref{Lemma:Ass4}, that there is a \(h\) for which the squeezing holds in expectation, as required by Assumption~\ref{as:four}.

\begin{lemma}\label{Lemma:bounds}
Assume that Properties~\ref{Prop:Finite}.1,~\ref{Prop:Finite}.2,~\ref{Prop:Finite}.4 and~\ref{Prop:Finite}.5 hold. Let \(U\) be a solution to ODE~\cref{eq:ODE} contained in the invariant set \(\mathscr{B} = B(0,K^{1/2})\). Then  \(\delta(t) = U(t) - u(t)\) satisfies 
\begin{equation}\label{eq:er_bound}
|\delta(t)|^2 \leq |\delta_n|^2e^{\kappa(t-t_n)},
\end{equation}
and 
\begin{equation}\label{eq:Pdelta_bound}
|P\delta|^2\leq |\delta_n|^2(a_4+ a_5\rho_n^{1/2})(t-t_n)+|P\delta(t_n)|^2,
\end{equation}
for \(t \in [t_n,t_{n+1})\), \(n \in \mathbb{N}_0\), \(\kappa = 2(2a_1K^{1/2}-1)\), \(a_4 =2e^{\kappa h}(\frac{a_1^2}{a_2} K+a_3) \), \(a_5 =2a_1 e^{3\kappa h/2}, \) and \(\rho_n\) is as in \cref{Lemma:Ass12}.
\end{lemma}

\begin{comment}
The proof is very similar to that in~\cite{Hayden2011}, Lemma 2.4 for the Lorenz~'63 model. We take the inner product of the error Equation \eqref{eq:error_eq} with \(\delta\) and use Properties \ref{Prop:Finite}.2, \ref{Prop:Finite}.4 and \ref{Prop:Finite}.5 as well as Gronwall's inequality.   
\end{comment}

\begin{comment}
\textbf{Proof of (32):}
Take the inner product of the error Equation~\eqref{eq:error_eq} with \(\delta\) and employ Properties~\ref{Prop:Finite}.2 and~\ref{Prop:Finite}.5 to obtain
   \[\frac{1}{2}\frac{d|\delta|^2}{dt} + |\delta|^2 \leq 2|(B(U,\delta),\delta)|.\]
   
By Cauchy-Schwartz and Property~\ref{Prop:Finite}.4 we get
\begin{equation*}
\frac{1}{2}\frac{d|\delta|^2}{dt} + |\delta|^2 \leq 2a_1|U||\delta|^2 \leq 2a_1K^{1/2}|\delta|^2.
\end{equation*}
Finally by Gronwall's inequality,
\[|\delta(t)|^2\leq e^{\kappa(t-t_n)}|\delta_n|^2,\] where \(\kappa = 2(2a_1K^{1/2}-1)\).  \qed 
\end{comment}
\textbf{Outline of Proof:}
Proof of \cref{eq:er_bound} is straightforward and similar to the proof given for the Lorenz system in \cite{Hayden2011}, so we omit it for brevity.

Proof of \cref{eq:Pdelta_bound}; Taking inner product of the error Equation~\cref{eq:error_eq} with \(P\delta\) and applying Inequality~\cref{eq:Property6}, we get
\[\frac{1}{2}\frac{d|P\delta|^2}{dt} + a_2|P\delta|^2 - a_3|\delta|^2 + 2(B(U,\delta),P\delta)-(B(\delta,\delta),P\delta)\leq0.\]

Inequality \cref{eq:Pdelta_bound} is obtained by applying Cauchy-Schwarz, Property~\ref{Prop:Finite}.4, Inequality \cref{eq:er_bound}, Young's and the a priori bound, which holds by \cref{Lemma:Ass12}, to the above and then applying Gronwall's lemma.

\begin{comment}
\off
By Cauchy-Schwarz, Property~\ref{Prop:Finite}.4 and Young's,
\begin{align*}
\frac{1}{2}\frac{d|P\delta|^2}{dt} + a_2|P\delta|^2 &\leq 2a_1|U||\delta||P\delta|+ a_1|\delta|^3+a_3|\delta|^2,\\
& \leq (\frac{4a_1^2}{4a_2} K+a_3)|\delta|^2+ \frac{2a_2}{2}|P\delta|^2+  a_1|\delta|^3.
\end{align*}
Hence by \eqref{eq:er_bound},
\[\frac{d|P\delta|^2}{dt} \leq 2(\frac{a_1^2}{a_2} K+a_3) |\delta_n|^2e^{\kappa(t-t_n)}+ 2a_1|\delta_n|^3e^{3\kappa(t-t_n)/2}.\]
Note that \(e^{\kappa(t-t_n)} < e^{\kappa h}\) for all \(t \in [t_n,t_{n+1} )\). Furthermore, by Lemma~\ref{Lemma:Ass12}, there exists a process \(\rho_n\) such that \(|\delta_n| \leq \rho_n^{1/2}\), which we use to replace in the $ |\delta_n|^3 $ in the second term of the RHS by the a priori bound.
Let \(a_4 =2e^{\kappa h}(\frac{a_1^2}{a_2} K+a_3) \) and \(a_5 =2a_1 e^{3\kappa h/2}, \) then we get
\begin{equation*}
\frac{d|P\delta|^2}{dt} \leq (a_4+ a_5\rho_n^{1/2})|\delta_n|^2.
\end{equation*}
Integrating from $ t_n $ to $t$ yields
\begin{equation*}
|P\delta|^2 \leq |\delta_n|^2(a_4+ a_5\rho_n^{1/2})(t-t_n)+|P\delta(t_n)|^2.
\end{equation*}
Recall that \(h, |\delta_n|^2\) and its a priori bound \(\rho_n\) are fixed at the beginning of the interval \([t_n,t_{n+1})\) and do not vary within it.
\off
\end{comment} 

The next lemma shows that Assumption~\ref{as:three} holds.
\begin{lemma}\label{Lemma:ass3}
Let \(U \in \mathscr{B} = B(0,K^{1/2})\) be a solution to ODE~\cref{eq:ODE}, satisfying Properties~\ref{Prop:Finite}.1 to~\ref{Prop:Finite}.5 with \(\delta(t)\) as defined by Equation~\cref{eq:error}, then there exist continuous functions \(M\) \(:\mathbb{R^{+}} \times \mathbb{R^{+}} \rightarrow \mathbb{R^{+}}\) and and \(\gamma\)\(:\mathbb{R^{+}} \rightarrow \mathbb{R^{+}}\) such that
\[|\delta(t)|^2 \leq M(t-t_n, \rho_n)|\delta_n|^2 + \gamma(t-t_n)|P\delta_n|^2,\] 
 for \(t\geq t_n\), where \[M(\tau, \rho_n) = e^{-\tau}(1 + a_6\int_0^{\tau}(a_4+ a_5\rho_n^{1/2})e^s s ds) \] and \[\gamma(\tau)= a_6(1-e^{-\tau}),\] with \(a_6= b^2K\).
\end{lemma}

\begin{proof}
Taking inner product of error Equation~\cref{eq:error_eq} with \(\delta\) and using Properties~\ref{Prop:Finite}.5 and~\ref{Prop:Finite}.2 we get
\begin{equation*}
\frac{1}{2}\frac{d|\delta|^2}{dt}+ |\delta|^2 \leq  2|(B(U,\delta),\delta)|.
\end{equation*}
Note that \(|U|\leq K^{1/2}\). Using \cref{Lemma:Prop1.3} and then Young's, we obtain
\begin{equation*}\label{eq:delta_bound}
\frac{1}{2}\frac{d|\delta|^2}{dt} + |\delta|^2 \leq  |\delta|^2/2 + b^2K|P\delta|^2/2,
\end{equation*}
and hence 
\begin{equation}\label{eq:L12}
\frac{d|\delta|^2}{dt} + |\delta|^2 \leq b^2K|P\delta|^2.
\end{equation}
We use the bound \cref{eq:Pdelta_bound} on \(|P\delta|^2\) from \cref{Lemma:bounds} and replace in above inequality to obtain
 \[\frac{d|\delta|^2}{dt} + |\delta|^2 \leq b^2K\Big(|\delta_n|^2(a_4+ a_5\rho_n^{1/2})(t-t_n)+|P\delta(t_n)|^2\Big).\] Multiplying by the integrating factor \(e^{t-t_n}\) and using Gronwall we get
\begin{align*}
|\delta|^2 &\leq |\delta_n|^2M_n(t-t_n,\rho_n) + |P\delta_n|^2\gamma(t-t_n),
\end{align*} 
\begin{comment}
\begin{align*}
|\delta|^2 &\leq |\delta_n|^2e^{-(t-t_n)} + b^2Ke^{-(t-t_n)}\int_{t_n}^{t}e^{(s-t_n)}\Big(|\delta_n|^2(a_4+ a_5\rho_n^{1/2})(s-t_n)+|P\delta(t_n)|^2\Big)ds,\\
&=|\delta_n|^2M_n(t-t_n,\rho_n) + |P\delta_n|^2\gamma(t-t_n),
\end{align*} 
\end{comment}
where \[M_n(\tau):=M(\tau, \rho_n) = e^{-\tau}(1 + a_6\int_0^{\tau}(a_4+ a_5\rho_n^{1/2})e^s s ds) \] and \[\gamma(\tau)= a_6(1-e^{-\tau}),\] with \(a_6= b^2K\). Since \(\rho_n\) is continuous w.r.t. \(\tau\) for all \(\tau >0\), so are \(M_n\) for a.e.\ \(\omega\).  \end{proof}

We note that in this case the \(\gamma_n\) are all the same, non-random and finite for all \(\tau \geq 0\). Therefore Assumption~\ref{as:four} is satisfied if the following lemma holds.
\begin{lemma}\label{Lemma:Ass4}
There exists  \(\tau^* > 0\) such that  \(\mathbb{E}M_n(\tau) < 1\) and \(\mathbb{E}\gamma_n(\tau) < \infty\) for all \(\tau \in (0,\tau^*]\).
\end{lemma}

\begin{proof} We wish to show that the function \[m(\tau) = \mathbb{E}M_n(\tau) =  e^{-\tau}(1 + a_6\int_0^{\tau}(a_4+ a_5\mathbb{E}(\rho_n^{1/2}))e^s s ds)\] 
is less than 1 in some neighbourhood around 0. The a priori bound  \(\rho_n\), and consequently \(M_n\), is not well defined at zero. However, we will show that \(\mathbb{E}(\rho_n(\tau)^{1/2})\tau^{1/2}\) is finite in a neighbourhood around \(\tau=0\), that is, \(\mathbb{E}(\rho_n(\tau)^{1/2})s^{1/2} < B\) for some constant \(B>0\), for all \(s\leq \tau\)  and \(\tau\) sufficiently small.

Supposing the above holds, we have that in this neighbourhood
\begin{equation} \label{eq:m}
m(\tau) \leq e^{-\tau}(1 + a_6\int_0^{\tau}(a_4s^{1/2}+a_5B)s^{1/2}e^s ds):= \overline{m}(\tau),
\end{equation}
which implies that
\[m(0) = \lim_{\tau \to 0} \mathbb{E}M_n(\tau) \leq \lim_{\tau \to 0}\overline{m}(\tau) = 1.\]
Furthermore, \[\frac{d\overline{m}(\tau)}{d\tau} = -\overline{m}(\tau)+a_6(a_5\tau^{1/2} + a_4B)e^{\tau}\tau^{1/2},\] and hence \[\frac{d\overline{m}(0)}{d\tau} = -1.\]
Therefore, there exists a \(\tau^*\) such that \(\overline{m}(\tau)<1\) for all \(0< \tau \leq \tau^*\). Hence by the bound in~\cref{eq:m} the same is true of \(m(\tau)\), for sufficiently small \(\tau\). 

It remains to show that \( \mathbb{E}(\rho_n^{1/2})\tau^{1/2} =\mathbb{E}((\rho_n\tau)^{1/2}) \) is bounded in a neighbourhood around \(\tau=0\).  Recall that
\begin{align*}
\rho_n & = \bar{K}+\frac{6|f|^2}{1-e^{-\tau}} + \frac{4|U(t_0)|}{\tau}+4\sigma^2\sum_{k=0}^{\infty}e^{-k\tau}|R_{n-k}|^2.
\end{align*}
Therefore,
\begin{align*}
\mathbb{E}((\rho_n\tau)^{1/2})& \leq \mathbb{E}(\rho_n\tau)^{1/2}\\
&=\Big(\bar{K}\tau+\frac{6|f|^2}{1-e^{-\tau}}\tau+4|U(t_0)|^2+4\sigma^2\tau\sum_{k=0}^{\infty}e^{-k\tau}\Big)^{1/2}\\
&=\Big(\bar{K}\tau+4|U(t_0)|^2+\frac{6|f|^2+4\sigma^2}{1-e^{-\tau}}\tau\Big)^{1/2}.
\end{align*} This bound is continuous at 0, and the limit is
\[\lim_{\tau \to 0}\mathbb{E}((\rho_n\tau)^{1/2})\leq (4|U(t_0)|^2+6|f|^2+4\sigma^2)^{1/2},\] which is finite. 
\end{proof}

Before turning to the N-S equations we will analyse two well known finite dimensional systems, known as Lorenz~'63 and~'96, that are commonly used as model problems for data assimilation. 
\subsection{Lorenz '63 model}\label{ss:L63}
The Lorenz~'63 model consists of a system of three coupled ODEs, obtained from the N-S equations by truncation of the Fourier series to the first three modes \cite{Lorenz1963,TemamBook}. It is given by
\[
\begin{cases}
\dot{U_1} = -\alpha U_1 + \alpha U_2,\\
\dot{U_2} = -\alpha U_1 - U_2 - U_1U_3,\\
\dot{U_3} = -bU_3 + U_1U_2 - b(r+\alpha) ,
\end{cases}
\]

where the parameters \(b, r, \alpha \geq 0\) are real constants with standard values of \(b = 10, r=8/3, \alpha = 28\).

We can write this system in the form of ODE~\cref{eq:ODE}, (see e.g.~\cite{FoiasJolly2001}), where \[A = \begin{pmatrix} \alpha & -\alpha & 0 \\ \alpha & 1 & 0 \\ 0 & 0 & b \end{pmatrix}, B(U,\bar{U}) = \begin{pmatrix} 0 \\ (U_1\bar{U_3}+U_3\bar{U_1})/2 \\ -(U_1\bar{U_2}+U_2\bar{U_1})/2 \end{pmatrix}, f=\begin{pmatrix} 0 \\ 0 \\ - b(r+\alpha) \end{pmatrix}.\] 
The standard observation operator \(P\) is the projection onto the \(U_1\) subspace. With this operator \(P\) all items of Property~\ref{Prop:Finite} are easily verified. Furthermore, we have \[(B(U,V),PW)=0\] for all \(U,V,W \in \mathbb{R}^3\), meaning that the nonlinear part of the flow is always perpendicular to the observations.

This last property is specific to Lorenz~'63; it does not hold for Lorenz~'96 or \\ N~-~S. It means that we can have a much simplified estimate for \(|P\delta|^2\), since taking inner product of the error Equation~\cref{eq:error_eq} and \(P\delta\) and applying Inequality~\cref{eq:Property6} now yields;
\[\frac{d|P\delta|^2}{dt} + 2a_2|P\delta|^2 \leq 2a_3|\delta_n|^2e^{\kappa(t-t_n)}\leq 2a_3|\delta_n|^2e^{\kappa h}.\] Setting \(a_7 = 2a_3e^{\kappa h}\), the estimate~\cref{eq:Pdelta_bound} on \(|P\delta|^2\) is simplified to
\off \[ |P\delta(t)|^2\leq e^{-2a_2(t-t_n)}(\frac{a_7}{2a_2}|\delta_n|^2(e^{-2a_2(t-t_n)}-1) + |P\delta_n|^2). \] \off
We note that the stochastic \(\rho_n\) no longer appears. We follow the proof of \cref{Lemma:ass3} till Equation~\cref{eq:L12} and then use the simplified bound obtained above. Thus we get,
\[|\delta|^2 \leq |\delta_n|^2M(t-t_n) + \gamma(t-t_n)|P\delta_n|^2,\]
where
\off \[M(\tau)= e^{-\tau}(1 + a_8\int_0^{\tau}e^{s}-e^{(-2a_2+1)s} ds) \] \off  and \[\gamma(\tau)= b^2Ke^{-\tau}\int_0^{\tau}e^{(-2a_2+1)s} ds, \]
where \off \(a_8 = b^2K\frac{a_7}{2a_2}.\) \off

We can see that in the particular case of Lorenz~'63, we get a stronger result because \(M\) is deterministic and does not depend on the size of \(|\delta(t_n)|\). Consequently we just need to show that the non-random function \(M(\tau)<1\) for Assumption~\ref{as:four} to hold. This can readily be verified as \(M(0) = 1\) and \off \(M^{'}(\tau) = -M(\tau) + a_8(e^{\tau}-e^{-2a_2 \tau})+ \kappa e^{-\tau}a_8\int_0^{\tau}e^{s}-e^{(-2a_2+1)s} ds\)\off , so that \(M^{'}(0) = -1 < 0\). Therefore, there exists a \(\tau^{*}>0\) such that \(M(\tau)<1\) for all \(\tau < \tau^{*}\).  

In this case the \(C_n\) of \cref{th:Main} have a much simpler form. Choose some \(h \in (0,\tau^*)\). Let \(\zeta > 0\) be some constant such that \(M(h) < \zeta < 1\). We can replace \(M_k\) by the constant \(\zeta\) in Equation~\cref{eq:Meq} and get;
\begin{align*}
|\delta_{n}|^2 &\leq \sigma^2\sum_{l=0}^{n-1} \zeta^l \gamma|R_{n-l-1}|^2 +\zeta^n|QU_0-\eta|^2+\sigma^2|R_n|^2.
\end{align*}
 Therefore \[\limsup_{n}\Big(|\delta_{n}|^2-\sigma^2(\sum_{l=0}^{\infty}\zeta^l \gamma|R_{n-l-1}|^2+ |R_n|^2)\Big) \leq  0.\]
Hence, a possible form of \(C_n\) is \(C_n =\sum_{l=0}^{\infty}\zeta^l \gamma|R_{n-l-1}|^2+ |R_n|^2 \), which is a stationary process due to the assumptions on \(R_n\). Furthermore, \(\mathbb{E}(C_n) = \frac{1-\zeta+ \gamma}{1-\zeta} < \infty\). Therefore, since \(C_n \geq 0\), it is a.s.\ finite. 

We note also that \[\limsup_{n}\mathbb{E}|\delta_{n}|^2\leq \sigma^2\mathbb{E}(C_n) = \frac{\sigma^2(1-\zeta+\gamma)}{1-\zeta},\]
so that the long-term mean square of the error is proportional to the strength of the noise, since constants \(\zeta\) and \(\gamma\) are independent of the noise and only depend on the data assimilation interval \(h\).

The bounding process \(C_n\) gives little information in the limit \(h \to 0\), because then \(\zeta~\to~1\). The same problem arises using 3DVAR as shown by~\cite{Law2014}, however they also give numerical results showing that the accuracy of the filter is fortunately a lot better than the theoretical bound implies. Clearly the bounds we give are not sharp since we make a number of estimates along the way. The main problem with our analysis for small \(h\) is that we are always summing the squared magnitude of the observational error. If \(h\) is small enough however, the dynamics is close to the identity, which should lead to considerable cancellations between the propagated errors. This is not taken into account in our approach. 

We remark also that the above \(P\) is not the only observation projection that would allow for \cref{th:Main} to hold. Any such \(P\) would need to satisfy Property~\ref{Prop:Finite}.3. That is, \(B(QU,QU) = 0\), so that the image of \(Q\) is contained in the null space of \(B\). The null space of \(B\) is given by \(U_3U_1=0\) and \(U_1U_2 =0\) so that it is composed of the plane \(U_1 = 0\) and the line \(U_3 = U_2= 0\). This means that \(Q\) must project either onto the \((U_2,U_3)\)-plane or the \(U_1 \) subspace or the origin. Since \(P = I-Q\), \(P\) can project either onto the \((U_2,U_3)\)-plane or the \(U_1\) subspace, or the whole space (i.e.~P is the identity). We note that observing only the \(U_2\) or only the \(U_3\) subspace would not work.

\subsection{Lorenz '96 model}\label{ss:L96}

The Lorenz~'96 model~\cite{Lorenz1996} is given by
\[\frac{dU_i}{dt} = (U_{i+1} - U_{i-2})U_{i-1}-U_i + F,\] for \(i=1...N\), \(N=3M\), for some \(M \in \mathbb{N}\) with \(U_{-1} = U_{N-1}, U_0 = U_N\), \(U_{N+1} = U_1\) and \(F=8\).

As given in~\cite{Law2016}, in this model \(A\) is the \(N \times N\) identity matrix, \(f=(8,...,8)^T\) is an \(N\) dimensional vector, and the symmetric bilinear form is given by \[B(U,V)_i = -\frac{1}{2}((U_{i+1}- U_{i-2})V_{i-1}+(V_{i+1}- V_{i-2})U_{i-1}).\]
The projection operator \(P\) is produced by setting every third column of the identity matrix to 0. That is, \[P = (e_1, e_2,0, e_4,..., 0, e_{N-2}, e_{N-1}, 0).\]
With the above observation operator it has been shown, see~\cite{Law2014}, that Property~\ref{Prop:Finite} holds and that \(a_2 = 1\) and \(a_3 = 0\) since $ A $ is the identity matrix. Furthermore, we have that \(b = 6\) and \(a_1 = 2\).

In some ways the Lorenz~'96 model behaves more like the 2D Navier-Stokes, in that the equation for $ P $ is not as simple; Lorenz~'63 is in this sense exceptional. Thus, in the case of Lorenz~'96 we cannot easily deduce an explicit form for the process \(C_n\).
\section{Application to Navier Stokes}
\label{sec:NS}
In this section we show that the 2-D incompressible Navier-Stokes equations, with L-periodic boundary conditions, satisfy Assumptions~\ref{as:one} to~\ref{as:five}, and therefore that \cref{th:Main} and \cref{th:Main2} hold also for this model.

As we will see, the strategy for showing Assumptions~\ref{as:three} and \ref{as:four} for the N-S equations will differ from the finite dimensional case we saw in \cref{sec:Finite}. In the case of N-S, we are able to use only the $ Q $ part of the error equation to derive the ``squeezing" property of Assumption~\ref{as:three}. This is due to the specific form that the observation operator \(P_{\lambda}\) takes, which means that the $ Q $ equation represents the higher modes, which are dissipated quicker, the larger the \(\lambda\). In the Lorenz models all modes are dissipated at the same rate so we cannot hope to adjust the operator \(P\) in order to obtain the same effect.

Following the notation of~\cite{Hayden2011}, let \(\Omega = [0,L] \times [0,L]\). The equations for the velocity field \(u\) and pressure \(p\) are given by
\begin{align}\label{eq:NS}
& \frac{\partial u}{\partial t} - \nu \Delta u + (u.\nabla)u +\nabla p = f,  \\ \nonumber
& \nabla \cdot u =0,  
\end{align}
where $ \nu $ is the kinematic viscosity and \(f\) the time independent body forcing. Let $ \mathbb{V} $ be the space of L-periodic trigonometric polynomials, with zero divergence and zero constant term. That is, \[\mathbb{V} = \{u: \mathbb{R}^2 \rightarrow \mathbb{R}^2; \text{L-periodic trig. polynomial}, \nabla . u= 0, \int_{\Omega} u = 0 \}, \] and let \(H\) be the closure of \(\mathbb{V}\) in \(L^2(\Omega)\) and \(V\) the closure of \(\mathbb{V}\) in Sobolev space \(H^1\). Let \(v \in \mathbb{V}\) and let \(u \in V\) be a solution to Equation~\cref{eq:NS}. Take the \(L^2\) inner product of~\cref{eq:NS} with \(v\) to get \[(\frac{\partial u}{\partial t},v) - \nu(\Delta u, v) + (u\cdot\nabla u,v) + (\nabla p,v) = (f,v).\] Since \(v\) is divergence-free we obtain for the pressure term \[(\nabla p,v) = \int_{\Omega} \nabla p \cdot v = -\int_{\Omega} p \nabla \cdot v=0,\] where we also use that \(v\) is periodic. By density of \(\mathbb{V} \in H^1\), the weak form
       \begin{equation}\label{eq:NS2}
       \frac{du}{dt} + \nu A u + B(u,u) = f
       \end{equation}
of the N-S equations holds for all \(v \in V\). Equation~\cref{eq:NS2} is an ODE in the dual space \(V^*\), so that $ A $ and $ B $ are operators from \(V\) to \(V^*\). If \(u \in H^2\) then \((Au,v) = \int_{\Omega} -\Delta u \cdot v \ dx \) and \((B(u,u),v) = \int_{\Omega} (u\cdot\nabla u)\cdot v \ dx \).

We can express \(u \in H\) by its Fourier series
\[ u= \sum_{\bar{k} \in \mathscr{J}} u_{\bar{k}} e^{i\bar{k}.x},\]
where \[ \mathscr{J} = \Big\{ \bar{k}=\frac{2 \pi}{L}(k_1, k_2): k_i \in \mathbb{Z}, \bar{k} \neq 0 \Big\}. \]      
We define norms on \(H, V\) and \(H^2 \cap H\) respectively as
\[|u|^2 = L^2\sum_{\bar{k} \in \mathscr{J}}|u_{\bar{k}}|^2,\]
\[\|u\|^2 = L^2\sum_{\bar{k} \in \mathscr{J}}\bar{k}^2|u_{\bar{k}}|^2,\] and 
\[|Au|^2 = L^2\sum_{\bar{k} \in \mathscr{J}}\bar{k}^4|u_{\bar{k}}|^2,\]

which can be shown to be equivalent to the standard norms on \(L^2, H^1\) and \(H^2\) on these spaces.

The key idea of the approach taken in~\cite{Hayden2011}, and which we follow, is that there is a natural splitting of the phase space \(V\) into a finite-dimensional sub-space and its infinite dimensional orthogonal complement such that the orthogonal projection of the solution onto the finite dimensional subspace dominates.

We define the orthogonal projection $P_\lambda$ as
\[P_\lambda u = \sum_{|\bar{k}|^2 \leq \lambda} u_{\bar{k}} e^{i\bar{k}.x}, \]
where $0 < \lambda \in \mathbb{Z}$. We say that \(P_\lambda\) is a projection onto the low modes.

Let us state some well known properties of the system. In this setting and with initial conditions in \(V\), the existence and uniqueness of strong solutions is shown for example in~\cite{Robinson2001}. Therefore we can define a semi-flow. We will verify Assumption~\ref{as:one} and~\ref{as:two} for Equation~\cref{eq:NS2} by the following Theorem which is proved in~\cite{Jones1992}.\\
\begin{theorem}
 Let \(u(t)\) solve the N-S Equation~\cref{eq:NS2} and \(u_0 \in H\), then the following estimate holds
\begin{equation} \label{eq:NS_disspi}
\|u(t)\|^2 \leq e^{-\nu\lambda_1(t-s)}\|u(s)\|^2 + \frac{1}{\nu}\int_{s}^{t}e^{-\nu\lambda_1(t-\tau)}|f|^2 \ d\tau
\end{equation}
for every \(0 <s\leq t\), where \(\lambda_1\) is the smallest eigenvector of \(A\). In particular, we have
\begin{equation}\label{eq:NS_attractor}
\limsup_{t \to \infty} \|u(t)\|^2 \leq \frac{|f|^2}{\nu^2\lambda_1}:=K.
\end{equation}
\end{theorem}
It follows from \cref{cor:apriori} that Assumption~\ref{as:two} is satisfied with constant \(c_1 = \nu \lambda_1\) and \(c_2 = K\).

It follows from Inequality~\cref{eq:NS_disspi} that the ball \(B(0,r)\) with \(r > K^{1/2}\) is an absorbing set because whatever bounded set we start with there will be a time after which it will be contained in the ball. Furthermore it's straightforward to show from~\cref{eq:NS_disspi} that \(B(0,r)\) is forward invariant, as required for Assumption~\ref{as:one}.

In the case where no noise is present in the observations, the existence of a function \(M\), as required for Assumption~\ref{as:three}, is shown in~\cite{Hayden2011}, Theorem 3.9. We follow the same reasoning but with the adjustment that in our setting \(P_\lambda \delta(t_n) \neq 0\), so that the induction argument used to ensure a bound on \(\|\delta_n\|^2\) is in our case impossible due to the noise term in the observation that can be arbitrarily large. Hence, we replace the \(R = \|\delta_0\|^2\) bound from~\cite{Hayden2011}, by an a priori bound from Assumption~\ref{as:two}. We conclude that whenever there exists a \(\rho>0\) such that \(\|\delta_n\|^2 < \rho\) we have for \(t \in [t_n,t_{n+1})\),
\[\|Q_\lambda \delta(t)\|^2 \leq M(t-t_n, \rho)\|\delta(t_n)\|^2, \] where
\[M(h) = e^{-\nu \lambda h} \Big(1+\int_{0}^{h}g(s, \rho) e^{\nu\lambda s} \ ds \Big)\] and \[g(s, \rho) = C_1 \lambda ^{1/4} e^{\kappa s}(\rho(h,\omega)^{1/2}e^{\kappa s /2}+ 2K^{1/2})^2 + C_2 e^{\kappa s}(\rho(h,\omega)^{1/2}e^{\kappa s /2}+ 2K^{1/2})^{8/3}, \] and where
\(C_1 = 2^{-1/4}\nu^{-1}\lambda_1^{-1/4}\), \(C_2 = 5^{5/3}2^{-22/3}3\nu^{-5/3}\lambda_1^{-1/3}\)\footnote{We have used the explicit value for the dimensionless constant \(c = 2^{-3/2}\) that appears in~\cite{Hayden2011}, Theorem 3.4.}. Further, \(K\) is the size of the attractor of the N-S dynamical system defined by Equation~\cref{eq:NS_attractor}. Finally, \(\kappa~=~2^{-1/3}(5/8)^{5/3}(3/8)\nu^{-5/3}\lambda_1^{-1/3}K^{4/3}\) is the constant as in~\cite{Hayden2011},~Theorem~3.8.

%Going back to the function \(M_k(h)\) and replacing the \(\|\delta(t_n)\|\) in the definition with the a priori bound \(\rho_n\) gives us a random variable on the space \(\Omega\); %

We want to use \cref{th:Main} to show that this random bound is sufficient to obtain convergence. Indeed, we can show that Assumption~\ref{as:four} holds.
\begin{theorem}
Suppose that $\mathbb{E}(|R_{0}|^{8/3})< \infty$, then for all \(h>0\), there exists a \(\lambda^* < \infty\) such that for all \(\lambda > \lambda^*\), Assumption~\ref{as:four} holds. That is, \(\mathbb{E}(M(h,\rho_0(h))) < 1\).
\end{theorem}

\begin{proof}
By the previous discussion, we have that 
\[\mathbb{E}(M(h,\rho_0(h))) = e^{-\nu \lambda h} \Big(1+\int_{0}^{h}\bar{g}(s,\rho_0(h)) e^{\nu\lambda s} \ ds \Big),\]
where  
\begin{equation*}
\bar{g}(s,\rho_0(h)) := C_1 \lambda^{1/4}e^{\kappa s}\mathbb{E}(l(h,s)^2) + C_2 e^{\kappa s}\mathbb{E}(l(h,s)^{8/3}),
\end{equation*}
and where \(l(h,s):=\rho_0(h)^{1/2}e^{\kappa s/2}+ 2K^{1/2}\).

Note that \(\bar{g}(s,\rho_0(h)) \leq \bar{g}(h, \rho_0(h))\) for all \( s \leq h\). Then 
\begin{align*}
\mathbb{E}(M(h,\rho_0(h))) & \leq e^{-\nu\lambda h} \Big(1+ \bar{g}(h, \rho_0(h)) \int_{0}^{h} e^{\nu \lambda s} \ ds \Big) \\
& = e^{-\nu\lambda h} +  \frac{\bar{g}(h, \rho_0(h))}{\nu \lambda} \Big(1- e^{- \nu \lambda h}\Big).
\end{align*}
From the above it follows that  \(\mathbb{E}(M_n(h))<1\) if
\(-\nu\lambda + \bar{g}(h,\rho_0(h)) < 0\).  Using the definition of \(\bar{g}\), we get
\begin{equation}\label{eq:lambda2}
-\nu \lambda + C_1 \lambda ^{1/4} e^{\kappa h}\mathbb{E}(l^2) + C_2 e^{\kappa h}\mathbb{E}(l^{8/3}) < 0,
\end{equation}
where \(l:=l(h,h)\).

It is clear that Inequality~\cref{eq:lambda2} will hold for some sufficiently large \(\lambda\) if the second and third terms of~\cref{eq:lambda2} are finite. It is sufficient to show that \(\mathbb{E}(l^{8/3})\) is finite, since then, any lower moment is finite. 

\begin{comment}To see this, consider any real \(\alpha > \beta >0\) and positive random variable \(X\). We can write \(X^\beta \leq 1+ X^\alpha\) so that \(\mathbb{E}(X^\beta) \leq 1+ \mathbb{E}(X^\alpha)\) by linearity and monotonicity of the integral. Therefore, if \(\mathbb{E}(X^\alpha)\) is finite, so is \(\mathbb{E}(X^\beta)\).\footnote{We can show this using Jensen's inequality also, since \(\alpha/\beta >1\); \[\mathbb{E}(X^\beta)^{\alpha/\beta} \leq \mathbb{E}(X^\alpha) < \infty. \]}\end{comment}

Recall that \(l^{8/3}= (\rho_0(h)^{1/2}e^{\kappa h /2}+ 2K^{1/2})^{8/3} \). It is sufficient to show that \newline \(\mathbb{E}(\rho_0(h)^{4/3})<~\infty\) since
\begin{align*}
\mathbb{E}(l^{8/3}) & = \int (\rho_0(h)^{1/2}e^{\kappa h /2}+ 2K^{1/2})^{8/3} d\mathbb{P} \\
& = \| (\rho_0(h)^{1/2}e^{\kappa h /2}+ 2K^{1/2})\|_{8/3}^{8/3}  \\
& \leq \Big(e^{\kappa h /2}\|\rho_0(h)^{1/2}\|_{8/3} + 2K^{1/2}\Big)^{8/3},\ 
\end{align*}
where in the last step we applied the Minkowski inequality. 

It's clear that the right hand side of the above inequality is finite if \[\|\rho_0(h)^{1/2}\|_{8/3} =\|\rho_0(h)\|_{4/3}^{1/2} < \infty.\]

\begin{comment}
It's clear that the right hand side of the above inequality is finite if \[\|\rho_0(h)^{1/2}\|_{8/3} = \Big(\int \rho_0(h)^{4/3} d\mathbb{P} \Big)^{3/8} < \infty.\] which is true if 
\[\int \rho_0(h)^{4/3} d\mathbb{P} < \infty.  \]
\end{comment}
Using the Minkowski inequality on the a priori bound we get
\[\|\rho_0(h)\|_{4/3} \leq \bar{K}+ F(h)+ 4\sigma^2\sum_{k=0}^{\infty} e^{-\nu\lambda_1kh}\|R_{-k}^2\|_{4/3}, \] 
where \(\bar{K}\) and  \(F(h)\) are both deterministic and the right hand side is finite if \(h>0\) and \(\|R_{-k}^2\|_{4/3} < \infty\). 
\end{proof}

The above result does not hold uniformly for small \(h\) since the bound diverges at \(h=~0\).

In the previous theorem we saw that for any \(h>0\), there exists a finite \(\lambda\) which guarantees that \(\mathbb{E}(M(h,\rho_0(h))) < 1\). We can compute an explicit expression for a possible \(\lambda\) from Equation~\cref{eq:lambda2}, which is given in \cref{Lemma:appendix} in the Appendix.

\subsection*{Acknowledgements}
We would like to thank Peter Jan van Leeuwen, Andrew M. Stuart, Edriss S. Titi for fruitful discussions.

%\printbibliography
%\bibliographystyle{siamplain}
%\bibliography{All_papers}

\section{Appendix}
\begin{lemma} \label{Lemma:appendix}
Equation~\cref{eq:lambda2} holds for all 
\begin{equation}\label{eq:lambda}
\lambda \geq \max\Big(2^{-1}e^{4/3\kappa h}\mathbb{E}(l^2)^{4/3},5^{5/3}2^{-19/3}3e^{\kappa h}\mathbb{E}(l^{8/3})\Big)\lambda_1^{-1/3}\nu^{-8/3}
\end{equation}
\end{lemma}

\begin{proof}
 We consider two possible cases of the second term of Inequality~\cref{eq:lambda2} being greater or smaller than the third term, which correspond to \(\lambda\) being greater or smaller than the expression
\begin{equation}\label{eq:tag1}
 \Big(\frac{\mathbb{E}(l^{8/3})}{\mathbb{E}(l^2)}\Big)^4(5^{5/3}2^{-19/12}3)^4\lambda_1^{-1/3}\nu^{-8/3}:=M_1.
\end{equation}
Replacing in Inequality~\cref{eq:lambda2}, we have that for \(\lambda \) greater  than or equal to~\cref{eq:tag1}, if \(\lambda\) holds for below equation then it holds for~\cref{eq:lambda2} as well;
\[-\nu \lambda + 2^{-3/4}\nu^{-1}\lambda_1^{-1/4}  \lambda ^{1/4} e^{\kappa h}\mathbb{E}(l^2) < 0,\] so that
\[\lambda >  2^{-1}\nu^{-8/3}\lambda_1^{-1/3}  e^{4/3\kappa h}\mathbb{E}(l^2)^{4/3}:=M_2,\] 
and hence
\begin{equation}\label{eq:lambda3}
\lambda >  \max\Big(M_1,M_2\Big).
\end{equation}
On the other hand, if 
\( \lambda \) is less than expression~\cref{eq:tag1}, we can replace Inequality~\cref{eq:lambda2} with
\[-\nu \lambda + 5^{5/3}2^{-19/3}3\nu^{-5/3}\lambda_1^{-1/3} e^{\kappa h}\mathbb{E}(l^{8/3}) < 0,\] so that
\[\lambda >  5^{5/3}2^{-19/3}3\lambda_1^{-1/3}\nu^{-8/3}e^{\kappa h}\mathbb{E}(l^{8/3}):=M_3,\] and hence
\begin{equation}\label{eq:lambda4}
M_3< \lambda <M_1. 
\end{equation}

There are solutions for \(\lambda\) in Inequality~\cref{eq:lambda4} if and only if \[e^{\kappa h} < \frac{\mathbb{E}(l^{8/3})^3}{\mathbb{E}(l^2)^4}5^52^{-16}3^3,\] so that
\[e^{4/3\kappa h} <\frac{\mathbb{E}(l^{8/3})^4}{\mathbb{E}(l^2)^{16/3}}(5^52^{-16}3^3)^{4/3}.\] 
Multiplying both sides by \(2^{-1}\nu^{-8/3}\lambda_1^{-1/3}               \mathbb{E}(l^2)^{4/3}\) we get precisely that
\[M_2 < M_1.\]

Conversely, when \(M3 > M_1\), we have that \(M_2>M_1\), which means that Inequality~\cref{eq:lambda3} becomes
\begin{equation}\label{eq:lambda5}
\lambda > M_2.
\end{equation}
Putting Inequalities~\cref{eq:lambda4} and~\cref{eq:lambda5} together, we see that we require that
\[\lambda > \max \Big(M_2, M_3\Big).\]
\end{proof}
\end{document}